\documentclass[11pt]{article}
\usepackage[margin=1in]{geometry} 
\usepackage{amssymb,amsfonts,amsmath,bbm,mathrsfs,stmaryrd}
\usepackage{xcolor}
\usepackage{url}

\usepackage[colorlinks,
             linkcolor=black!75!red,
             citecolor=blue,
             pdftitle={Top_gen},
             pdfauthor={Michael Brannan},
             pdfproducer={pdfLaTeX},
             pdfpagemode=None,
             bookmarksopen=true
             bookmarksnumbered=true]{hyperref}

\usepackage{tikz}
\usetikzlibrary{arrows,calc,decorations.pathreplacing,decorations.markings,intersections,shapes.geometric,through,fit,shapes.symbols,positioning,decorations.pathmorphing}

\usepackage{braket}

\usepackage[amsmath,thmmarks,hyperref]{ntheorem}
\usepackage{cleveref}

\crefname{section}{Section}{Sections}
\crefformat{section}{#2Section~#1#3} 
\Crefformat{section}{#2Section~#1#3} 

\crefname{subsection}{\S}{\S\S}
\crefformat{subsection}{#2\S#1#3} 
\Crefformat{subsection}{#2\S#1#3} 

\theoremstyle{plain}

\newtheorem{lemma}{Lemma}[section]
\newtheorem{proposition}[lemma]{Proposition}
\newtheorem{corollary}[lemma]{Corollary}
\newtheorem{theorem}[lemma]{Theorem}

\theoremstyle{nonumberplain}
\newtheorem{theoremN}{Theorem}

\theoremstyle{plain}
\theorembodyfont{\upshape}
\theoremsymbol{\ensuremath{\blacklozenge}}

\newtheorem{definition}[lemma]{Definition}
\newtheorem{example}[lemma]{Example}
\newtheorem{remark}[lemma]{Remark}
\newtheorem{convention}[lemma]{Convention}

\crefname{definition}{definition}{definitions}
\crefformat{definition}{#2definition~#1#3} 
\Crefformat{definition}{#2Definition~#1#3} 

\crefname{ex}{example}{examples}
\crefformat{example}{#2example~#1#3} 
\Crefformat{example}{#2Example~#1#3} 

\crefname{remark}{remark}{remarks}
\crefformat{remark}{#2remark~#1#3} 
\Crefformat{remark}{#2Remark~#1#3} 

\crefname{convention}{convention}{conventions}
\crefformat{convention}{#2convention~#1#3} 
\Crefformat{convention}{#2Convention~#1#3}

\crefname{lemma}{lemma}{lemmas}
\crefformat{lemma}{#2lemma~#1#3} 
\Crefformat{lemma}{#2Lemma~#1#3} 

\crefname{proposition}{proposition}{propositions}
\crefformat{proposition}{#2proposition~#1#3} 
\Crefformat{proposition}{#2Proposition~#1#3} 

\crefname{corollary}{corollary}{corollaries}
\crefformat{corollary}{#2corollary~#1#3} 
\Crefformat{corollary}{#2Corollary~#1#3} 

\crefname{theorem}{theorem}{theorems}
\crefformat{theorem}{#2theorem~#1#3} 
\Crefformat{theorem}{#2Theorem~#1#3} 

\crefname{assumption}{assumption}{Assumptions}
\crefformat{assumption}{#2assumption~#1#3} 
\Crefformat{assumption}{#2Assumption~#1#3} 

\crefname{equation}{}{}
\crefformat{equation}{(#2#1#3)} 
\Crefformat{equation}{(#2#1#3)}

\theoremstyle{nonumberplain}
\theoremsymbol{\ensuremath{\blacksquare}}

\newtheorem{proof}{Proof.}

\newcommand\bC{\mathbb C}

\newcommand\bN{\mathbb N}

\newcommand\cA{\mathcal A}
\newcommand\cB{\mathcal B}
\newcommand\cC{\mathcal C}

\newcommand\cG{\mathcal G}

\newcommand\cI{\mathcal I}

\newcommand\cO{\mathcal O}
\newcommand\cP{\mathcal P}

\newcommand\bb[1]{{\mathbb #1}}

\DeclareMathOperator{\id}{id}

\newcommand\numberthis{\addtocounter{equation}{1}\tag{\theequation}}

\newcommand{\cl}[1]{\mathcal{#1}}

\title{Bigalois extensions and the graph isomorphism game}
\author{Michael Brannan\footnote{Department of Mathematics, Texas A\&M University,  \url{mbrannan@math.tamu.edu}}, Alexandru Chirvasitu\footnote{Department of Mathematics, University at Buffalo,  \url{achirvas@buffalo.edu}}, Kari Eifler\footnote{Department of Mathematics, Texas A\&M University,  \url{keifler@math.tamu.edu}}, Samuel Harris\footnote{Deparment of Pure Mathematics, University of Waterloo, \url{sj2harri@uwaterloo.ca}}, \\Vern Paulsen\footnote{Deparment of Pure Mathematics \& The Institute for Quantum Computing, University of Waterloo, \url{vpaulsen@uwaterloo.ca}}, Xiaoyu Su\footnote{Department of Mathematics, Texas A\&M University, \url{xiaoyuegongzi@math.tamu.edu}}, Mateusz Wasilewski\footnote{Department of Mathematics, Katholieke Universiteit Leuven, \url{mateusz.wasilewski@kuleuven.be}}}

\begin{document}

\date{\today}

\maketitle

\begin{abstract}
We study the graph isomorphism game that arises in quantum information theory.  We prove that the non-commutative algebraic notion of a quantum isomorphism between two graphs is same as the more physically motivated one arising from the existence of a perfect quantum strategy for graph isomorphism game.  This is achieved by showing that every algebraic quantum isomorphism between a pair of (quantum) graphs $X$ and $Y$ arises from a certain measured bigalois extension for the quantum automorphism groups $G_X$ and $G_Y$ of $X$ and $Y$.   In particular, this implies that the quantum groups $G_X$ and $G_Y$ are monoidally equivalent.  We also establish a converse to this result, which says that a compact quantum group $G$ is monoidally equivalent to the quantum automorphism group $G_X$ of a given quantum graph $X$ if and only if $G$ is the quantum automorphism group of a quantum graph that is algebraically quantum isomorphic to $X$.  Using the notion of equivalence for non-local games, we apply our results to other synchronous games, including the synBCS game and certain related graph homomorphism games.
%We study the graph isomorphism game that arises in quantum information theory from the perspective of bigalois extensions of compact quantum groups.  We show that every algebraic quantum isomorphism between a pair of (quantum) graphs $X$ and $Y$ arises as a quotient of a certain measured bigalois extension for the quantum automorphism groups $G_X$ and $G_Y$ of the graphs $X$ and $Y$.  In particular, this implies that the quantum groups $G_X$ and $G_Y$ are monoidally equivalent.  We also establish a converse to this result, which says that every compact quantum group $G$ monoidally equivalent to $G_X$ is of the form $G_Y$ for a suitably chosen quantum graph $Y$ that is quantum isomorphic to $X$.  As an application of these results, we deduce that the $\ast$-algebraic, C$^\ast$-algebraic, and quantum commuting (qc) notions of a quantum isomorphism between classical graphs $X$ and $Y$ all coincide.  Using the notion of equivalence for non-local games, we apply our result to other synchronous non-local games, including the synBCS game and certain related graph homomorphism games.

\end{abstract}


\vspace{.5cm}


%\tableofcontents

%%%%%%%%%%%%%%%%%%%%%%%%%%%%%%%%%%%%%%%%%%%%%%%%%%%%%%%%%%%%%%%%%%%%%%%%%%%%%
%%%%%%%%%%%%%%%%%%%%%%%%%%%%%%%%%%%%%%%%%%%%%%%%%%%%%%%%%%%%%%%%%%%%%%%%%%%%%

\section{Introduction}

Finite input-output games have received considerable attention in the quantum information theory literature as tools for investigating the structure of {\it quantum correlations}. The latter are meant to model the following setup (in one of several ways, depending on the specific chosen model \cite{tsi}):

The perennial experimenters Alice and Bob, share an entangled quantum state, each performs one of $n$ quantum experiments and returns one of $k$ outputs. The respective quantum correlation is then defined as the collection of conditional probabilities $(p(a,b|x,y))$ of obtaining a given pair of outputs $a, b$ for a given pair of inputs, $x, y$. 

To recast this in game theoretic terms one typically proceeds as follows. 
Alice (A) and Bob (B) are regarded as cooperating players trying to supply ``correct" answers to a referee (R) who communicates with A/B via sets of questions (inputs) $I_A,I_B$. Alice and Bob then reply with answers from their respective sets of outputs $O_A,O_B$. The game rules, which are known to A, B, and R, are embodied by a function $\lambda:O_A \times O_B \times I_A \times I_B \to \{0,1\}$, such that A and B win a round of the game if their respective replies $a$ and $b$ to the questions $x$ and $y$ satisfy $\lambda(a,b|x,y) = 1$, and they lose the round otherwise.  Prior to the round, A and B can cooperate to develop a winning strategy, but are not allowed to communicate once the game begins. Cooperation consists of a pre-arranged strategy, which can be deterministic or random.  We identify probabilistic strategies with the collection of conditional probability densities $(p(a,b|x,y))$ that they produce.

 A probabilistic strategy is called {\it perfect} or {\it winning}, if the probability that they give an incorrect pair of answers is 0, i.e., $\lambda(x,y,a,b)=0 \implies p(a,b|x,y)=0$. The key point is that the set of such conditional probability densities $(p(a,b|x,y))$ that can be obtained via quantum experiments in an entangled state is larger than the set of densities that can be obtained from classical shared randomness.  For this reason {\it winning} quantum strategies can often be shown to exist when classical ones do not exist.

%What gives the subject its quantum flavor is the class of probabilistic strategies being considered: A and B are placed in an entangled quantum state, and then perform respective quantum experiments based on the questions $x\in I_A, y \in I_B$ they receive.  

Our starting point in this context is the {\it graph isomorphism game} introduced in \cite{amrssv}. Given two finite graphs $X$ and $Y$ this game has inputs that are the disjoint union of the vertices of $X$ and the vertices of $Y$ and outputs that are the same set.  The referee sends A, B each a vertex $x_A,x_B \in V(X) \cup V(Y)$ and receives respective answers $y_A, y_B \in V(X) \cup V(Y)$. Winning conditions require that
\begin{itemize}
\item $x_A$ and $y_A$ belong to different graphs;
\item ditto for $x_B$ and $y_B$;
\item the ``relatedness'' of the inputs is reflected by that of the outputs: $x_A=x_B$, or are distinct and connected by an edge, or distinct and disconnected if and only if the same holds for the $y$ vertices.  
\end{itemize}
It turns out that the game has a perfect deterministic strategy if and only if the two graphs are isomorphic. This observation motivates the authors of \cite{Lu17} to introduce a notion of {\it quantum isomorphism} between finite graphs, relating to the existence of less constrained, random strategies for the graph isomorphism game. Whether or not two finite graphs $X$ and $Y$ are quantum-isomorphic is governed by a $*$-algebra $\cl A(Iso(X,Y))$, a non-commutative counterpart to the function algebra of the space of isomorphisms $X\to Y$. More precisely, the algebra of continuous functions on the space of isomorphisms $X\to Y$ can be recovered as the abelianization of $\cl A(Iso(X,Y))$.

One is thus prompted to consider ``quantum spaces'' in the sense of non-commutative geometry: $*$-algebras or C$^*$-algebras, thought of as function algebras on the otherwise non-existent spaces. In the same spirit we will work with {\it quantum graphs} (finite-dimensional C$^*$-algebras equipped with some additional structure mimicking an ``adjacency matrix'' \Cref{qset-graph}) and {\it quantum groups} i.e. (objects dual to) non-commutative $*$-algebras with enough structure to resemble algebras of representative functions on compact groups (\Cref{subse.qg}).

As is the case classically, every quantum graph $X$ has a {\it quantum automorphism group} $G_X$. The recent papers \cite{Lu17, MuRuVe18, MuRuVe18b} uncover further remarkable connections between graph isomorphism games and quantum automorphism groups. Moreover, while \cite{Lu17} focuses on classical graphs, \cite{MuRuVe18, MuRuVe18b} consider a more general categorical quantum mechanical framework which leads naturally to the notion of a quantum graphs and the generalization of the graph isomorphism game to that framework.

In particular, \cite{MuRuVe18b} obtains a characterization of (finite-dimensional) quantum isomorphic quantum graphs $X$, $Y$ in terms of simple dagger Frobenius monoids in the category of finite dimensional representations of the Hopf $\ast$-algebra $\cO(G_X)$ of the corresponding quantum automorphism group $G_X$.  On the other hand, \cite{Lu17} uses ideas from quantum group theory to establish the equivalence between the existence of C$^\ast$-quantum isomorphisms for graphs and the existence of perfect strategies for the isomorphism game within the so-called {\it quantum commuting} framework.

Here, we continue in the same vein investigating connections between quantum automorphism groups of graphs and the graph isomorphism game, taking a  somewhat dual approach to the one in \cite{MuRuVe18, MuRuVe18b}:

We regard the game $\ast$-algebra $\cA(Iso(X,Y))$ as indicated above, as a non-commutative analogue of the space of isomorphisms $X \to Y$.  In particular, we say that $X$ and $Y$ are {\it algebraically quantum isomorphic}, and simply write $X \cong _{A^*} Y$, if $\cA(Iso(X,Y)) \ne 0$.  Classically, the space of isomorphisms between two graphs is a principal homogeneous bundle over the automorphism groups of both graphs.  In other words, if $Iso_c(X,Y)$ denotes the space of graph isomorphisms $X \to Y$ and $\text{Aut}(X)$ (resp.  $\text{Aut}(Y)$) denotes the automorphism group of $X$ (resp. $Y$), then the canonical left/right actions $\text{Aut}(Y)\curvearrowright Iso_c(X,Y) \curvearrowleft\text{Aut}(X)$ are free and transitive. One of our main results is the quantum analogue of this remark; it is \Cref{non-zero} below, and can be paraphrased as follows:

\begin{theoremN}
Let $X$ and $Y$ be two quantum graphs. If the quantum isomorphism space $\cl A(Iso(X,Y))$ is non-trivial then it is a quantum principal bi-bundle (bigalois extension) over the quantum automorphism groups $G_X$ and $G_Y$ of $X$ and $Y$ respectively. 
\end{theoremN}

This (non-commutative) bundle-theoretic perspective on $\cA(Iso(X,Y))$ has advantages: although the construction of $\cA(Iso(X,Y))$ is purely algebraic and does not assume the existence of any C$^\ast$-representations of this object, we use the above result to show that this algebra always admits a faithful invariant state whenever it is non-zero (cf. \Cref{traces-exist}), leading to connections with the notion of  monoidal equivalence between quantum automorphism groups.  Loosely speaking, we say that two compact quantum groups are {\it monoidally equivalent} if their  categories of finite-dimensional unitary representations are equivalent as rigid C$^\ast$-tensor categories.  Our main result here is an amalgamation of \Cref{traces-exist} and \Cref{stability-monoidal}, and says:

\begin{theoremN}
Let $X$ be a quantum graph and $G_X$ its quantum automorphism group.  Then the following hold: 
\begin{enumerate}
\item If $Y$ is another quantum graph such that $X \cong_{A^*}Y$, then $\cA(Iso(X,Y))$ admits a faithful state and $G_X$ is monoidally equivalent to the quantum automorphism group $G_Y$.  If both $X$ and $Y$ are moreover classical graphs, then  $\cA(Iso(X,Y))$ admits a faithful tracial state.
\item Conversely, for any compact quantum group $G$ monoidally equivalent to $G_X$, one can construct from this monoidal equivalence a quantum graph $Y$, an isomorphism of quantum groups $G \cong G_Y$, and an algebraic quantum isomorphism $X \cong_{A^*}Y$.    
\end{enumerate}  
\end{theoremN}

Recasting all of the above in the context of the (classical) graph isomorphism game, our results show that the condition $\cA(Iso(X,Y)) \ne 0$ is sufficient to ensure the existence of perfect quantum strategies for this game (\Cref{th:qiso-monoidal} and \Cref{th.cls-grph}):

\begin{theoremN}
Two classical graphs $X$ and $Y$ are algebraically quantum isomorphic if and only if the graph isomorphism game has a perfect quantum-commuting $(qc)$-strategy. 
\end{theoremN}

A weaker version of the above theorem (assuming the existence of a non-zero C$^\ast$-algebra representation of $\cA(Iso(X,Y))$) was recently proved in \cite{Lu17}.

Using a notion of $\ast$-equivalence for non-local games, we also use the above result to deduce the existence of perfect $qc$-strategies from purely algebraic data for {\it synchronous binary constraint system (syncBCS)} games and certain related {\it graph homomorphism games} (\Cref{vs-app1} and \Cref{vs-app2}).  We find all these results very striking because for generic synchronous games (e.g. the graph homomorphism game) the $\ast$-algebra $\cA(\mathcal G)$ governing the game may be non-zero even if this algebra has no C$^\ast$-representations (and hence no perfect quantum strategies) \cite{HMPS17}.

We note also that the above results could be recast in a much broader context -- rather than quantum graphs, we could consider arbitrary {\it finite quantum structures}: quantum sets (i.e. finite-dimensional $C^*$-algebras with a fixed state) equipped with arbitrary tensors. Input-output isomorphism games could then be constructed as in the case of graphs, and the discussion replicated in that general framework. Our focus on (quantum) graphs is motivated by the contingent fact that the latter have received considerable interest in the literature.

After recalling some preliminary material in \Cref{se.prel} and generalities on quantum sets and Galois extensions in \Cref{se.qs} we prove our main results in \Cref{se.bigal}. Finally, \Cref{se.appl} contains further applications to finite input-output games.

%%%%%%%%%%%%%%%%%%%%%%%%%%%%%%%%%%%%%%%%%%%%%%%%%%%%%%%%%%%%%%%%%%%%%%%%%%%%%
%%%%%%%%%%%%%%%%%%%%%%%%%%%%%%%%%%%%%%%%%%%%%%%%%%%%%%%%%%%%%%%%%%%%%%%%%%%%%
\section{Preliminaries}\label{se.prel}
\subsection{Some notation}  If $n$ is a natural number, we sometimes write $[n]$ for the ordered set $\{1, 2, \ldots, n\}$.  All vector spaces considered here are over the complex field. We use the standard leg numbering notation for linear operators on tensor products of vector spaces.  For example, if $X,Y,Z$ are vector spaces and $T:X\otimes Y \to X \otimes Y$ is a linear map, then $T_{13}:X \otimes Z \otimes Y \to X \otimes Z \otimes Y$ is the linear map which acts as $T$ on the first and third leg of the triple tensor product, and as the identity on the second leg.  We also typically denote the identity map on a vector space by $\iota$.

%%%%%%%%%%%%%%%%%%%%%%%%%%%%%%%%%%%%%%%%%%%%%%%%%%%%%%%%%%%%%%%%%%%%%%%%%%%%%

\subsection{Games and strategies}

We lay out some definitions and a few basic properties of games and strategies. We will primarily be concerned with the graph isomorphism game, the graph homomorphism game and two versions of a game based on solving systems of linear equations over the binary field.

By a {\it two-person finite input-output game} we mean a tuple $\cl G=(I_A, I_B, O_A, O_B, \lambda)$ where $I_A, I_B, O_A, O_B$ are finite sets and
\[ \lambda: I_A \times I_B \times O_A \times O_B \to \{ 0,1 \} \] is a function that represents the rules of the game, sometimes called the predicate.  The sets $I_A$ and $I_B$ represent the inputs that the players Alice and Bob can receive, and the sets $O_A$ and $O_B$, represent the outputs that Alice and Bob can produce, respectively.  A referee selects a pair $(v,w) \in I_A \times I_B$, gives Alice $v$ and Bob $w$, and they then produce outputs (answers), $a \in O_A$ and $b \in O_B$, respectively.  They win the game if $\lambda(v,w,a,b) =1$ and loose otherwise.  Alice and Bob are allowed to know the sets and the function $\lambda$ and cooperate before the game to produce a strategy for providing outputs, but while producing outputs, Alice and Bob only know their own inputs and are not allowed to know the other person's input. Each time that they are given an input and produce an output is referred to as a {\it round} of the game.

We call such a game {\it synchronous} provided that: (i) Alice and Bob have the same input sets and the same output sets, which we denote by $I$ and $O$, respectively, and (ii) $\lambda$ satisfies:
\[ \forall v \in I, \,\, \lambda(v,v,a,b) = \begin{cases} 0 & a \ne b\\ 1 & a=b \end{cases},\] that is, whenever Alice and Bob receive the same inputs then they must produce the same outputs. To simplify notation we write a synchronous game as $\cl G= (I,O, \lambda)$.

A {\it deterministic strategy} for a game is a pair of functions, $h: I_A \to O_A$ and $k:I_B \to O_B$ such that if Alice and Bob receive inputs $(v,w)$ then they produce outputs $(h(v), k(w))$. A deterministic strategy  wins every round of the game if and only if
\[ \forall  (v,w) \in I_A \times I_B,  \qquad  \lambda(v,w, h(v), k(w)) =1.\] Such a strategy is called a {\it perfect deterministic strategy}. It is not hard to see that for a synchronous game, any perfect deterministic strategy must satisfy, $h=k$.  

On the other hand, a strategy for a game is called {\it random} if it can happen that for different rounds of the game, when Alice and Bob receive the input pair $(v,w)$ they may produce different output pairs.
A {\it random strategy} thus yields a conditional probability density $p(a,b|v,w)$, which represents the probability that, given inputs $(v,w) \in I_A \times I_B$, Alice and Bob produce outputs $(a,b) \in O_A \times O_B$.  Thus, $p(a,b|v,w) \ge 0$ and for each $(v,w),$
\[ \sum_{a \in O_A, b \in O_B} p(a,b|v,w) =1.\] 

In this paper we identify random strategies with their conditional probability densities, so that a random strategy will simply be a conditional probability density $p(a,b|v,w)$.

A random strategy is called {\it perfect} if
\[ \lambda(v,w,a,b)=0 \implies p(a,b|v,w) =0, \, \forall (v,w,a,b) \in I_A \times I_B \times O_A \times O_B.\]
Thus, for each round of the game, a perfect strategy gives a winning output with probability 1.

Given a particular set of conditional probability densities, one can ask if the game not only has a perfect random strategy, but has one that belongs to a particular set of densities.
The different kinds of probability densities that are studied in this context generally fall into two types: There are the {\it local (loc)} densities, also called the {\it classical} densities, which arise from ordinary random variables defined on probability spaces, and then there are the {\it quantum} densities that arise from the random outcomes of, especially, entangled quantum experiments.  However, there are several different mathematical models for describing the densities obtained from quantum experiments. These models lead to sets of conditional probability densities know variously as the {\it quantum (q), quantum spatial (qs)} (or sometimes {\it quantum tensor}), {\it quantum approximate (qa)}, and {\it quantum commuting (qc)} models.

Rather than go into a long explanation of the definitions of each of these sets, which is done many other places, we refer the reader to \cite{KPS18}, for their definitions and merely summarize some of their basic relations below. Given $n$ inputs and $k$ outputs, we denote the set of conditional probability densities $p(a,b|v,w)$ that belong to each of these sets by  $C_t(n,k)$, where $t$ can be $loc, q, qs, qa$ or $qc$.  The following containments are known:
\[ C_{loc}(n,k) \subseteq C_q(n,k) \subseteq C_{qs}(n,k) \subseteq C_{qa}(n,k) \subseteq C_{qc}(n,k).\]
Moreover, for $n,k \ge 2$, it is known that $C_{loc}(n,k) \ne C_{q}(n,k)$. While for $n \ge 5, k \ge 2$, we have $C_{qs}(n,k) \ne C_{qa}(n,k)$ by \cite{DPP}, and for
$n \ge 5, k \ge 3$, we have $C_q(n,k) \ne C_{qs}(n,k)$ \cite{CS}. The most famous question is whether or not $C_{qa}(n,k) = C_{qc}(n,k), \, \forall n,k \ge 2$, since this is known to be equivalent to {\it Connes' embedding conjecture}, first posed in \cite{con-cj}; see \cite{Oz13}.

We shall say that a game has a {\it perfect t-strategy} provided that it has a perfect random strategy that belongs to one of these sets, where $t$ can be either $loc, q, qs, qa$ or $qc$. Moreover, we work with even broader classes of strategies we term C$^*$ and $A^*$ (the latter being the broadest, i.e. weakest class; see \Cref{def.strtg}).

%%%%%%%%%%%%%%%%%%%%%%%%%%%%%%%%%%%%%%%%%%%%%%%%%%%%%%%%%%%%%%%%%%%%%%%%%%%%%
\subsection{The $\ast$-algebra of a synchronous game}
\label{sec:alg-sync-game}
In \cite{OP} a $\ast$-algebra was affiliated with the graph homomorphism game, $Hom(X,Y)$, whose representation theory determined whether or not a perfect t-strategy existed (see Section \ref{graphgame} for definitions). Later in \cite{P-lecnotes} and \cite{HMPS17, KPS18} these ideas were extended to any synchronous game.  We begin by recalling the $\ast$-algebra of a synchronous game and summarizing these results.  This $\ast$-algebra is defined by generators and relations arising from the rule function of the game.

%%%%%%%%%%%%%%%%%%%%%%%%%%%%%%%%%%%%%%%%%%%%%%%%%%%%%%%%%%%%%%%%%%%%%%%%%%%%%
%\subsection{The generators and relations}

Let $\cl G= ( I,O, \lambda)$ be a synchronous game and assume that the cardinality of $I$ is $|I|=n$ while the cardinality of $O$ is $|O|=m$. We will often identify $I$ with $\{ 0,..., n-1 \}$ and $O$ with $\{ 0,..., m-1 \}$. We let $\mathbb{Z}_m^{*n}$ denote the free product of $n$ copies of the cyclic group of order $m$ and let $\bb C[\mathbb{Z}_m^{*n}]$ denote the complex $\ast$-algebra of the group.  We regard the group algebra as a $\ast$-algebra, where for each group element $g$ we have $g^* = g^{-1}$.
%and %regard it
%as an (incomplete) inner product space,
%with the group elements forming an orthonormal set and the inner product  given by
%\[ \langle f, h \rangle = \tau(fh^*),\]
%where $\tau$ is the trace functional.

For each $v \in I$ we have a unitary generator $u_v \in \bb C[\mathbb{Z}_m^{*n}]$ such that $u_v^m = 1$. If we set $\omega = e^{2 \pi i/m}$ then the eigenvalues of each $u_v$ is the set $\{ \omega^a: 0 \le a \le m-1 \}$.  The ``orthogonal projection" onto the eigenspace corresponding to $\omega^a$ is given by
\begin{equation*}
\label{eq:efromu}
e_{v,a} = \frac{1}{m} \sum_{k=0}^{m-1} \big( \omega^{-a} u_v \big)^k,
\end{equation*}
and these satisfy
\[ 1 = \sum_{a=0}^{m-1} e_{v,a} \text{ and } u_v = \sum_{a=0}^{m-1} \omega^a e_{v,a}.\]
The set $\{ e_{v,a}: v \in I, 0 \le a \le m-1 \}$ is another set of generators for $\bb C[\mathbb{Z}_m^{*n}]$.

We let $\cl I(\cl G)$  \index{$\cl I(\cl G)$}
denote the 2-sided $\ast$-ideal in $\bb C[\mathbb{Z}_m^{*n}]$ generated by the set
\[ \{ e_{v,a}e_{w,b} \, | \ \lambda(v,w,a,b)=0 \} \]
and refer to it as {\it the ideal of the game $\cl G$}. We define the {\it  $\ast$-algebra of $\cl G$} to be the quotient
\[ \cl A(\cl G) = \bb C[\mathbb{Z}_m^{*n}]/\cl I(\cl G) .\] Note that since  $\lambda(v,v,a,b) =0, \forall v, a \ne b$, in the quotient we will have that $e_{v,a} e_{v,b} =0$.

It is not hard to see that if we, alternatively, started with the free $\ast$-algebra generated by $\{ e_{v,a}: v \in I, \, 0 \le a \le m-1 \}$ and formed the quotient by the two-sided ideal generated by:
\begin{itemize}
\item $e_{v,a} - e_{v,a}^*, \, \forall v,a$,
\item $e_{v,a} - e_{v,a}^2, \, \forall v,a$,
\item $1- \sum_a e_{v,a}, \, \forall v$
\item $e_{v,a}e_{w,b}, \, \forall v,w,a,b$ such that $\lambda(v,w,a,b) =0$,
\end{itemize}
then we obtain the same $\ast$-algebra. We are not asserting that this algebra is non-zero. In fact, it can be the case that the identity belongs to the ideal, in which case the algebra is zero.

The following is a summary of the results obtained in \cite{HMPS17} and \cite{KPS18} and illustrates the importance of this algebra.

\begin{theorem}[\cite{HMPS17, KPS18}]
\label{thm:intro1}
Let $\cl G=(I,O, \lambda)$ be a synchronous game.
\begin{enumerate}
\item $\cl G$ has a perfect deterministic strategy if and only if $\cl G$ has a perfect loc-strategy if and only if there exists a unital $\ast$-homomorphism from $\cl A(\cl G)$ to $\bb C$.
\item $\cl G$ has a perfect q-strategy if and only if $\cl G$ has a perfect qs-strategy if and only if there exists a unital $\ast$-homomorphism from $\cl A(\cl G)$ to $B(H)$ for some non-zero finite dimensional Hilbert space.
\item $\cl G$ has a perfect qa-strategy if and only if there exists a unital $\ast$-homomorphism of $\cl A(\cl G)$ into an ultrapower of the hyperfinite $II_1$-factor,
\item $\cl G$ has a perfect qc-strategy if and only if there exists a unital C$^\ast$-algebra $\cl C$ with a faithful trace and a unital $\ast$-homomorphism $\pi: \cl A(\cl G) \to \cl C$.
\end{enumerate}
\end{theorem}

This theorem motivates the following definitions.
\begin{definition}\label{def.strtg}
  Let $\cl G$ be a synchronous game.  We say that $\cl G$ has a {\it perfect A$^\ast$-strategy} provided $\cl A(\cl G)$ is non-zero, and we say that $\cl G$ has a {\it perfect C$^\ast$-strategy} provided that there is a unital $\ast$-homomorphism from $\cl A(\cl G)$ into $B(H)$ for some non-zero Hilbert space $H$.
\end{definition}

%%%%%%%%%%%%%%%%%%%%%%%%%%%%%%%%%%%%%%%%%%%%%%%%%%%%%%%%%%%%%%%%%%%%%%%%%%%%%

\subsection{Graphs and related games} \label{graphgame}

A {\it graph} $X$ is specified by a vertex set $V(G)$ and an edge set $E(X) \subseteq V(X) \times V(X)$, satisfying $(v,v) \notin E(X)$ and $(v,w) \in E(X) \implies (w,v) \in E(X)$.
%The \it{c-coloring game} for $X$ has inputs $I_A=I_B = V(X)$ and outputs $O_A=O_B = \{ 1,...,c \}$ where the outputs are thought of as different colors. They win provided that whenever Alice and Bob receive adjacent vertices, i.e., $(v,w) \in E$, their outputs are different colors and when they receive the same vertex they must output the same color. Thus,
%$(v,w) \in E(X) \implies \lambda(v,w,a,a)=0, \, \forall a$, $\lambda(v,v,a,b)=0, \, \forall v \in V(X), \, \forall a \ne b$ and the rule function is equal to 1 for all other tuples.  It is easy to see that this is a synchronous game.
Given two graphs $X$ and $Y$, a {\it graph homomorphism from X to Y} is a function $f:V(X) \to V(Y)$ with the property that $(v,w) \in E(X) \implies (f(v), f(w)) \in E(Y)$. We write $X \to Y$ to indicate that there exists a graph homomorphisms from $X$ to $Y$.  Graph homomorphisms encapsulate many familiar graph theoretic parameters.  If we let $K_c$ denote the {\it complete graph} on $c$ vertices, i.e., the graph where every pair of vertices is connected by an edge, then
\begin{itemize}
\item the {\it chromatic number} of $X$ is
$\chi(X) = \min \{ c: \exists \, X \to K_c \}$,
\item the {\it clique number} of $X$ is
$\omega(X) = \max \{ c: \exists \, K_c \to X \}$,
\item the {\it independence number} of $X$ is, $\alpha(X)= \max \{ c: \exists \, K_c \to \overline{X} \},$ 
\end{itemize}
where $\overline{X}$ denotes the {\it graph complement} of $X$, i.e., the graph whose edge set is the complement of  $X$'s.
 
The {\it graph homomorphism game} from $X$ to $Y$, which we shall denote by $Hom(X,Y)$, is a synchronous game with inputs $I_A=I_B = V(X)$ and outputs $O_A=O_B = V(Y).$  Alice and Bob win a round provided that whenever they receive inputs that are an edge in $X$, then their outputs are an edge in $Y$ and that whenever Alice and Bob receive the same vertex in $X$ they produce the same vertex in $Y$. This is also a synchronous game.

Note that a perfect deterministic strategy for the graph homomorphism game from $X$ to $Y$ is a function $h: V(X) \to V(Y)$ that is a graph homomorphism. In particular, a perfect deterministic strategy exists if and only if $\exists X \to Y$. Similarly, we say that there is a {\it t-homomorphism} from $X$ to $Y$ and write $X \stackrel{t}{\to} Y$ if and only if there exists a perfect t-strategy for the graph homomorphism game from $X$ to $Y$ for $t=q$, $qs$, etc.

%Finally, it is not difficult to see that if $K_c$ denotes the complete graph on $c$ vertices then a graph homomorphism exists from $X$ to $K_c$ if and only if $X$ has a c-coloring. This is because any time $(v,w) \in E(X)$ then a graph homomorphism must send them to distinct vertices in $K_c$. Indeed, the rule function for the c-coloring game is exactly the same as the rule function for the graph homomorphism game from $X$ to $K_c$.

\subsubsection{The graph isomorphism game}

Two graphs $X$ and $Y$ are {\it isomorphic} if and only if there exists a one-to-one onto function $f:V(X) \to V(Y)$ such that $(v,w)$ is an edge in $X$ if and only if $(f(v), f(w))$ is an edge in $Y$. We write $X \simeq Y$ to indicate that $X$ and $Y$ are isomorphic. If we let $A_X$ denote the adjacency matrix of $X$ and analogously for $A_Y$, then it is well-known and easy to check that $X \simeq Y$ if and only if there is a permutation matrix $P$ such that $A_X P = P A_Y$.

The {\it graph isomorphism game, Iso(X,Y)} between $X$ and $Y$ is a game with the property that two graphs are isomorphic if and only if there exists a perfect deterministic strategy for $Iso(X,Y)$. It was introduced by Atserias et al. \cite{amrssv}.

The easiest way to describe the rules for this game is in terms of the {\it relation} between a pair of vertices.  Formally, the relation on a graph is a function $rel: V(X) \times V(X) \to \{ 0, 1, -1 \}$ with
\begin{itemize}
\item $rel(v,w) =0 \iff v=w$,
\item $rel(v,w) =-1 \iff (v,w) \in E(X)$,
\item $rel(v,w) = +1 \iff v \ne w \text{ and } (v,w) \notin E(X)$.
\end{itemize}

We remark that the matrix $S_X := (rel(v,w))_{v,w \in V(X)}$ is known as the {\it Seidel adjacency matrix} of the graph.

The rules for this game can be stated loosely as requiring that to win, outputs must come from different graphs than inputs, outputs must have the same relations as inputs, and whenever one player's output is the same as the other player's input, then the same must hold for the other player. This final rule makes a deterministic strategy be a function and its inverse, instead of just a pair of functions.
The input set and output set for this game is the disjoint union of $V(X)$ with $V(Y)$ and 
\[\lambda: (V(X) \cup V(Y)) \times (V(X) \cup V(Y)) \to \{ 0,1 \},\] satisfies $\lambda(v,w,x,y)=1$ if and only if the following conditions are met:
\begin{itemize}
\item $x$ belongs to a different graph than $v$ and $y$ belongs to a different graph than $w$,
\item if $v$ and $w$ are both vertices of the same graph, then $rel(v,w) = rel(x,y)$.
\item if $v$ and $w$ are from different graphs and $x=w$, then $y=v$,
\item if $v$ and $w$ are from different graphs and $y=v$, then $x=w$.
\end{itemize}

Now it is not hard to see that this game is synchronous and it has a perfect deterministic strategy if and only if $X \simeq Y$.  Indeed, if it has a perfect deterministic strategy, then there must be a function $f: V(X) \cup V(Y) \to V(X) \cup V(Y)$ and the rules force $v \in V(X) \implies f(v) \in V(Y)$ and $x \in V(Y) \implies f(x) \in V(X)$. Denoting the restrictions of $f$ to $V(X)$ and $V(Y)$ by $f_1: V(X) \to V(Y)$ and $f_2: V(Y) \to V(X)$. The fact that $rel(v,w) = rel (f_1(v), f_1(w))$ tells us that $f_1$ is one-to-one and preserves the edge relationships, since $f_2$ is also one-to-one,  $card(V(X)) = card(V(Y))$ and so both $f_1$ and $f_2$ define graph isomorphisms.  However, note that the rules of the game do not require that $f_1$ and $f_2$ be mutual inverses.

We will write $X \simeq_t Y$ if and only if this game has a perfect t-strategy for $t \in \{ loc, q, qa, qc, C^*, A^* \}$.

The following result characterizes $\cl A(Iso(X,Y))$.

\begin{proposition} \label{Iso-alg}Let $X=(V(X), E(X))$ and $Y=(V(Y), E(Y))$ be graphs on $n$ vertices. Then $\cl A(Iso(X,Y))$ is generated by $4n^2$ self-adjoint idempotents $\{ e_{v,w}: v,w \in V(X) \cup V(Y) \}$ satisfying:
\begin{enumerate}
\item $e_{g,g^{\prime}} =0, \, \forall g, g^{\prime} \in V(X)$ and $e_{h,h^{\prime}} =0, \, \forall h, h^{\prime} \in V(Y)$,
\item $e_{g,h}^2= e_{g,h}^* = e_{g,h}, \, \forall g \in V(X), h \in V(Y),$
\item for $g \in V(X)$ and $h \in V(Y)$, $e_{g,h} = e_{h,g}$,
\item $\sum_{h \in V(Y)} e_{g,h} = 1, \, \forall g \in V(X)$,
\item $\sum_{g \in V(X)} e_{g,h} = 1, \, \forall h \in V(Y),$
\item $e_{g,h}e_{g,h^{\prime}} =0, \forall h \ne h^{\prime}$,
\item $e_{g,h}e_{g^{\prime},h} =0, \, \forall g \ne g^{\prime}$,
\item $\sum_{g^{\prime} : (g,g^{\prime}) \in E(X)} e_{g^{\prime}, h} = \sum_{h^{\prime}: (h,h^{\prime}) \in E(Y)} e_{g, h^{\prime}}, \, \forall g, h$.
\end{enumerate}
\end{proposition}
\begin{proof} Recall that  for any game, we will have generators,  $e_{x,y},  \, x,y \in V(X) \cup V(Y)$ with $e_{x,y}^2 = e_{x,y}^* = e_{x,y}$, $\sum_y e_{x,y} =1$, and 
$e_{x,y} e_{x,w}=0$ for $y \ne w$.
So (2) and (6) are automatically met.

To see (1), note that if $g, g^{\prime} \in V(X)$, then  $\lambda(g, x, g^{\prime}, y) =0,$ for all $x,y$. Hence for and fixed $x$ we have that 
\[ e_{g, g^{\prime}} = e_{g,g^{\prime}} \big( \sum_y e_{x,y} \big) = \sum_y e_{g,g^{\prime}}e_{x,y} =0.\] The case that $h, h^{\prime} \in V(Y)$ is identical.

Note that (4) follows from (1).

To see (3), note that
\[ e_{h,g} = e_{h,g}( \sum_{k \in V(Y)} e_{g,k}) = \sum_{k \in V(Y)} e_{h,g} e_{g,k}.\]
Now $\lambda(h,g,g,k)=0$ unless $h=k$, so we have that $e_{h,g} = e_{h,g}e_{g,h}$ 
A similar calculation shows that $e_{g,h} = e_{g,h}e_{h,g}$. Hence,
$e_{g,h} = e_{g,h}^* = (e_{g,h}e_{h,g})^*= e_{h,g} e_{g,h} = e_{h,g}$.

Now (5) follows from (3) and (4). Similarly, (7) follows from (3) and (6).

Finally to see (8), we have that
\begin{multline*}
\sum_{g^{\prime} : (g,g^{\prime}) \in E(X)} e_{g^{\prime}, h} = \big( \sum_{g^{\prime} : (g,g^{\prime}) \in E(X)} e_{g^{\prime}, h} \big) \big( \sum_{h^{\prime} \in V(Y)} e_{g,h^{\prime}} \big) = \\
\sum_{g^{\prime}, h^{\prime}: (g, g^{\prime}) \in E(X), h^{\prime} \in V(Y)} e_{g^{\prime}, h} e_{g,h^{\prime}} = \sum_{g^{\prime}, h^{\prime}:(g,g^{\prime}) \in E(X), (h,h^{\prime}) \in E(Y)} e_{g^{\prime},h} e_{g,h^{\prime}},
\end{multline*}
since $\lambda(g^{\prime}, g, h, h^{\prime})=0$ unless $(h,h^{\prime}) \in E(Y)$.
Similarly, one shows that $\sum_{h^{\prime}: (h,h^{\prime}) \in E(Y)} e_{g, h^{\prime}} $ is equal to this latter sum and (8) follows.
\end{proof}

\begin{remark}\label{rem-iso} A nice compact way to represent the above relations is to consider the $n \times n$ matrix $U=( e_{g,h})_{g \in V(X), h \in V(Y)}$. Then by (2) every entry is a self-adjoint idempotent, while (4) and (5) imply that $U^*U= UU^*$ is the identity matrix, i.e., that $U$ is a unitary. We also, by (6) and (7), have that entries in each row and column are pairwise ``orthogonal", i.e., have pairwise 0 product. Such a matrix $U$ will be referred to as a {\it quantum permutation} over the $\ast$-algebra $\cA (Iso(X,Y))$.

Equation (8) implies that
$(1 \otimes A_X)U = U(1 \otimes A_Y)$ where $A_X$ and $A_Y$ denote the adjacency matrices of the graphs, and $1$ is the unit of the algebra.   Thus, Proposition~\ref{Iso-alg} can be summarized as saying that $\cl A(Iso(X,Y))$ is the $\ast$-algebra generated by $\{ e_{g,h} : g \in V(X), h \in V(Y) \}$ subject to the relations that $U= ( e_{g,h})$ is a quantum permutation with $(1 \otimes A_X)U= U(1 \otimes A_Y)$.  We have that $X \simeq_{A^*} Y$ if and only if a non-trivial $\ast$-algebra exists satisfying these relations.
\end{remark}

\begin{remark} Combining Proposition~\ref{Iso-alg} with Theorem~\ref{thm:intro1}, we see that given two graphs $X$ and $Y$ on $n$ vertices:
\begin{itemize}
\item  $X \simeq_{q} Y$ if and only if there exist a $d$ and projections $E_{g,h} \in M_d$ such that $U= (E_{g,h})$ is a unitary in $M_n(M_d)$ and $(1 \otimes A_X)U=U(1 \otimes A_Y),$
\item $X \simeq_{qa} Y$ if and only if there exist projections $E_{g,h} \in \cl R^{\omega}$ such that $U=(E_{g,h}) \in M_n(\cl R^{\omega})$ is a unitary and $(1 \otimes A_X)U=U(1 \otimes A_Y),$
\item $X \sim_{qc} Y$ if and only if there exists projections $E_{g,h}$ in some C$\ast$-algebra $\cl A$ with a trace such that $U=(E_{g,h}) \in M_n(\cl A)$ is a unitary and $(1 \otimes A_X)U=U(1 \otimes A_Y)$,
\item $X \simeq_{C^*} Y$ if and only if there exists projections $E_{g,h}$ on a Hilbert space $H$ such that $U=(E_{g,h}) \in M_n(B(H))$ is a unitary and $(1 \otimes A_X) U = U (1 \otimes A_Y)$.
\end{itemize}
Also, if there exists a unital $\ast$-homomorphism from $\pi: \cl A(Iso(X,Y)) \to \mathbb C$, then $(\pi(e_{g,h})) \in M_n$ will be a permutation matrix, satisfying $A_X ( \pi(e_{g,h})) = (\pi(e_{g,h})) A_Y$, which is the classical notion of isomorphism for graphs.
\end{remark}

Note that we have the following obvious implications. 

\[ X \cong Y \implies X \cong_{q}Y  \implies X \cong_{qa} Y \implies X \cong_{qc}Y \implies X \cong_{C^*}Y \implies X \cong_{A^*} Y.\]

Moreover, it is known that the first two implications are not reversible \cite{amrssv,KPS18}.  The question of whether the third implication holds is still open.  The question whether the implications $X \cong_{A^*} Y \implies X \cong_{C^*} Y \implies X \cong_{qc}Y$ hold for generic $X$ and $Y$ had remained open for quite some time.  Only very recently the implication $C^* \implies qc$ was obtained in \cite{Lu17}.  One of our main results is that the implication $A^* \implies qc$ holds.  In other words, $\cA(Iso(X,Y)) \ne 0$ if and only if $\cA(Iso(X,Y))$ admits a tracial state.  This is somehow surprising, because the same conclusion cannot be made for the algebras $\cA(Hom(X,Y))$ \cite{HMPS17}.

%\subsubsection{LBCS Game}

%There is another quantum game Players $A$ and $B$ can enjoy, which we'll introduce in this section.

%\begin{definition}
%A {\it linear binary constraint system (LBCS)} $\cF$ is a family of binary variables $x_1,x_2,\ldots,x_n \in %\{-1,1\}$ and constraints of the form

%\[ \cC_l = \Bigg( \Pi_{x_i \in S_l}x_i = (-1)^{b_l} \text{ for } S_l \subseteq \{ x_1,\ldots,x_n\}, b_l \in \bF_2 \Bigg) %\]
%\end{definition}

%The LBCS game starts with the referee giving Player $A$ a constraint followed by giving Player $B$ a variable in her constraint. Player $A$ is to respond with values for all variables in her constraint and Player $B$ assigns a value to his variable. The game is won if the values Players $A$ and $B$ give are the same.

%\begin{theorem}
%There is a perfect quantum strategy in the tensor product framework if and only if there exists a finite dimensional operator solution to the multiplicative LBCS.
%\end{theorem}

%\begin{definition}
%An {\it operator solutions} is some self-adjoint operators $X_i$, with $X_i^2 = \id$ on a Hilbert space $H$ such that all operators appearing in a constraint commute with each other (so that the order of multiplication is clear).

%The solution is {\it finite dimensional} if $H$ is finite dimensional.
%\end{definition}

%%%%%%%%%%%%%%%%%%%%%%%%%%%%%%%%%%%%%%%%%%%%%%%%%%%%%%%%%%%%%%%%%%%%%%%%%%%%%
\subsection{Compact quantum groups}\label{subse.qg} 

We follow the references \cite{Wor98, DiKo94, Ti08, NeTu13} for the basics on (C$^\ast$-algebraic) compact quantum groups.

We begin by recalling that a {\it Hopf algebra} is a quadruple $(A, \Delta, S, \epsilon)$ where  $A$ is a unital associative algebra with multiplication map $m$, and  $\Delta:A \to A \otimes A$, $S:A \to A^{op}$, $\epsilon:A \to \bC$ are unital algebra morphisms satisfying
\begin{enumerate}
\item  $(\iota \otimes \Delta) \Delta = (\Delta \otimes \iota)\Delta$ (co-associativity).
\item $m(\iota \otimes S)\Delta = \epsilon(\cdot)1 = m(S \otimes \iota)\Delta$
\item $(\epsilon \otimes \iota) \Delta = (\iota \otimes \epsilon)\Delta = \iota$.
\end{enumerate}
The maps $\Delta, S,\epsilon$ given above are called the {\it comultiplication, counit, and antipode}, respectively.  We typically just refer to a Hopf algebra with the symbol $A$ if the other structure maps $m, \Delta, S, \epsilon$ are understood and there is no danger of confusion.  A {\it Hopf $\ast$-algebra} is a Hopf algebra $A$ where $A$ is a $\ast$-algebra and the comultiplication and counit are $\ast$-homomorphisms.

The following definition is (essentially) taken from \cite{DiKo94, Bi99} and is one of many equivalent ones. 
\begin{definition}\label{def.cqg}
  A {\it CQG algebra} is a Hopf $\ast$-algebra $A$ for which there exists a C$^\ast$-norm $\|\cdot\|$ on $A$ making the comultiplication $\Delta:A \to A \otimes A$ continuous with respect to the minimal C$^\ast$-tensor norm $\otimes_{\min}$ (in short, we say that $\|\cdot\|$ is {\it $\Delta$-compatible}).

A {\it compact quantum group (CQG)} is the object dual to a CQG algebra, i.e. we regard compact quantum groups and CQG algebras as mutually opposite categories. We write $G$ for a quantum group and $\cO(G)$ for its corresponding CQG algebra. 
\end{definition}

Since there is an identification between the objects of the categories of quantum groups and CQG algebras, we will on occasion abuse language and conflate the two.

The motivating example of a CQG is given by the Hopf $\ast$-algebra $\cO(G)$ of representative functions on a compact group $G$.  Here, $\Delta:\cO(G) \to \cO(G \times G) =\cO(G) \otimes \cO(G)$ is the map $\Delta f(s,t) = f(st)$, $Sf(t) = f(t^{-1})$, and $\epsilon(f) = f(e)$, where $e \in G$ is the unit.  Here the C$^\ast$-norm on $\cO(G)$ is the uniform norm coming from $C(G)$, and it is relatively easy to see that this is the unique C$^\ast$-norm making the comultiplication $\Delta$ $\otimes_{\min}$-continuous.

Motivated by the above example, it is customary to use the symbol $G$ to to denote an arbitrary CQG and write $A = \cO(G)$ for the Hopf $\ast$-algebra associated to $G$.  Here we are viewing $A$ as a non-commutative algebra of ``representative functions'' on some ``quantum space'' $G$, which comes equipped with a group structure.   

Another example of a CQG is given by the Pontryagin dual $\widehat{\Gamma}$ of a discrete group $\Gamma$.  Here $\cO(\widehat{\Gamma}) = \bC \Gamma$, $\Delta (\gamma) = \gamma\otimes \gamma$, $S \gamma = \gamma^{-1}$, and $\epsilon(\gamma) = 1$ for each $\gamma \in \Gamma$.  In this case, one can in general choose from a variety of $\Delta$-compatible C$^\ast$-norms.  The two most common ones are the maximal C$^\ast$-norm on $\bC \Gamma$ and the reduced C$^\ast$-norm, the latter being induced by the left regular representation of $\bC\Gamma$ on  $\ell^2\Gamma$.

A few ``purely quantum'' examples follow. 

\begin{example}\label{ex.u}
  Wang and Van Daele's {\it universal unitary quantum group} $U^+_F$ \cite{DaeWan96} associated to a matrix $F \in \text{GL}_n(\bC)$ is given by 
  \[
    \cO(U^+_F) = *\text{-algebra}\big( u_{ij}, \ 1 \le i,j \le n \ \big| \ u = [u_{ij}] \quad \& \quad (1 \otimes F)[u_{ij}^*] (1 \otimes F^{-1}) \text{ are unitary in }M_n(\cO(U^+_F))   \big)\]
together with Hopf-$\ast$-algebra maps $\Delta(u_{ij}) = \sum_k u_{ik}\otimes u_{kj}$, $S(u_{ij}) = u_{ji}^*$ and $\epsilon(u_{ij}) = \delta_{ij}$.

The term {\it universal} in the above definition will be made precise in the paragraph following Remark \ref{comod}.  For the time being, now we suffice it  to say that the quantum groups play the analogous universal role for compact matrix quantum groups that the ordinary compact matrix groups: Any compact matrix group $G$ arises as a closed quantum subgroup of some $U^+_F$.  That is, there exists a surjective Hopf $\ast$-algebra morphism $\cO(U_F^+) \to \cO(G)$. 
\end{example}

\begin{example}\label{ex.s}
  The {\it quantum permutation group} $S^+_n$ on $n$ points \cite{Wan98} is given by underlying Hopf $\ast$-algebra $A=\cO(S^+_n)$ which is the universal $*$-algebra generated by the entries of an $n\times n$ {\it magic unitary}: a matrix
  \begin{equation*}
    [u_{ij}]_{i,j} \in M_n(A)
  \end{equation*}
  consisting of self-adjoint projections summing up to $1$ across all rows and columns, and satisfying the orthogonality relations $u_{ij}u_{ik} = \delta_{jk}u_{ij}$ and $u_{kj}u_{ij} =  \delta_{ki}u_{ij}$.  The Hopf $\ast$-algebra maps $\Delta,S,\epsilon$ are defined exactly as for $U^+_F$.
\end{example}

Every compact quantum group $G$ comes equipped with a unique {\it Haar state}, which is a faithful state $h:\cO(G) \to \bC$ satisfying the left and right invariance conditions
\[
(\iota \otimes h)\Delta = h(\cdot)1 = (h \otimes \iota)\Delta.
\]
The norm on $\cO(G)$ induced by the GNS construction with respect to $h$ is always a $\Delta$-compatible norm, and it is the minimal such C$^\ast$-norm.  
We denote by $C_r(G)$, the corresponding C$^\ast$-algebra (the {\it reduced C$^\ast$-algebra of $G$}).  The {\it universal C$^\ast$-algebra of $G$}, $C^u(G)$, is the enveloping C$^\ast$-algebra of $\cO(G)$.  By universality, this C$^\ast$-norm is also $\Delta$-compatible.  

\begin{remark}
Often in the literature compact quantum groups are defined in terms of a pair $(A, \Delta)$ where $A$ is a unital C$^\ast$-algebra and $\Delta:A \to A \otimes_{\min}A$ is a co-associative unital $\ast$-homomorphism such that $\Delta(A)(1 \otimes A)$ and $\Delta(A)(A \otimes 1)$ are linearly dense in $A \otimes_{\min}A$.  One then obtains the Hopf $\ast$-algebra $\cO(G)$ as a certain dense $\ast$-subalgebra (spanned by coefficients of unitary representations of $G$, which we describe below).     
\end{remark}

Let $G $ be a CQG and $H$ a finite dimensional Hilbert space.  A {\it representation} of $G$ on $H$ is an invertible element $v \in \cO(G) \otimes B(H)$ such that 
\[
(\Delta \otimes \iota)v = v_{13}v_{23}.
\]
A representation of $G$ is called {\it unitary} if $v \in \cO(G) \otimes B(H)$ is unitary.  Note that if we fix an orthonormal basis $(e_i)_{i=1}^d$ for $H$, then a representation $v \in \cO(G)\otimes B(H)$ corresponds to an invertible matrix $v = [v_{ij}] \in M_n(\cO(G))$ such that
\[\Delta(v_{ij}) = \sum_{k=1}^d v_{ik} \otimes v_{kj} \qquad (1\le i,j \le d).\]

For any CQG, we always have the {\it trivial representation} on $\bC$ given by $v= 1 \in \cO(G) = \cO(G) \otimes B(\bC)$.  Given two representations $u \in \cO(G) \otimes B(H)$ and $v \in \cO(G) \otimes B(H)$, we can always form the {\it direct sum} $u \oplus v \in \cO(G) \otimes B(H \oplus K)$, {\it tensor product} $u \otimes v: = u_{12}v_{13} \in \cO(G) \otimes B(H \otimes K)$, and {\it conjugate representation} $\bar u \in \cO(G) \otimes B(\bar H)$ given by $\bar u = [u_{ij}^*]$ (if $u = [u_{ij}]$). A {\it morphism} between $u$ and $v$ is a linear map $T:H \to K$ such that $u(1 \otimes T) = (1 \otimes T)v$.  The Banach space of all morphisms between $u$ and $v$ is denoted by $\text{Mor}(u,v)$.  If $u$ and $v$ are unitary representations, then $T \in \text{Mor}(u,v) \iff T^* \in \text{Mor}(v,u)$.  We say that two representations $u$ and $v$ are {\it equivalent} if there exists an invertible element $T \in \text{Mor}(u,v)$.  We say that $u$ is {\it irreducible} if $\text{Mor}(u,u) = \bC 1$.

The fundamental theorem on finite dimensional representations of CQGs is stated as follows.

\begin{theorem}[\cite{Wor98}]
  Let $G$ be a CQG.  Every finite dimensional representation of $G$ is equivalent to a unitary representation, and every finite dimensional unitary representation of $G$ is equivalent to a direct sum of irreducible representations.  Moreover, $\cO(G)$ is linearly spanned by the matrix elements of irreducible unitary representations of $G.$
\end{theorem}

\begin{remark} \label{comod}
In the language of Hopf algebras, a (unitary) representation of $G$ is typically called a {\it (unitary) comodule} over $\cO(G)$.  These notions obviously make sense for general Hopf $\ast$-algebras. 
\end{remark}

We end this section by recalling that a matrix {\it Hopf $\ast$-algebra} is a Hopf $\ast$-algebra that is generated by the coefficients of some corepresentation $w= [w_{ij}] \in M_n(A)$ of $A$. A useful fact in this regard from \cite{DiKo94} is that if a matrix Hopf $\ast$-algebra $A$ is generated by a corepresentation $w$ that is equivalent to a unitary one, then $A = \cO(G)$ is the Hopf $\ast$-algebra of some compact quantum group $G$.  In this case, we call $G$ a {\it compact matrix quantum group} and we call $w$ a {\it fundamental representation} of $G$.   By replacing $w$ with an equivalent unitary representation $v$, note that $\cO(G)$ is still generated by the matrix elements of $v \in M_n(\cO(G))$, and $\bar v$ is a unitarizable representation.  Hence there exists some $F \in GL_n(\bC)$ so that $(1 \otimes F)\bar v(1 \otimes  F^{-1})$ is a unitary representation.  This means that there is a surjective morphism of Hopf $\ast$-algebras $\pi:\cO(U^+_F) \to \cO(G)$ defined by $(\pi \otimes \iota)u = v$ ,where $u$ is the fundamental representation of $U^+_F$ given  in its definition.  In particular, $G$ is a so-called {\it closed quantum subgroup} of $U^+_F$ (written $G < U^+_F$).

%%%%%%%%%%%%%%%%%%%%%%%%%%%%%%%%%%%%%%%%%%%%%%%%%%%%%%%%%%%%%%%%%%%%%%%%%%%%%
%%%%%%%%%%%%%%%%%%%%%%%%%%%%%%%%%%%%%%%%%%%%%%%%%%%%%%%%%%%%%%%%%%%%%%%%%%%%%
\section{Quantum sets, graphs and their quantum automorphism groups}\label{se.qs}

The examples of CQGs that feature in this paper are the quantum automorphism groups of certain finite structures, such as sets, graphs, and their quantizations.  In order to describe these objects, we first quantize the notion of a (measured) finite set, then proceed to quantum graphs.  All of the definitions that follow are quite standard in the operator algebra literature \cite{Wan98, ban-sym, Ba02, DeVa10}. The idea of a quantum set or a quantum graph also appears in \cite{MuRuVe18, MuRuVe18b} using the language of special symmetric dagger Frobenius algebras.

%%%%%%%%%%%%%%%%%%%%%%%%%%%%%%%%%%%%%%%%%%%%%%%%%%%%%%%%%%%%%%%%%%%%%%%%%%%%%
\subsection{Quantum sets and graphs}\label{qset-graph}

\begin{definition} \label{qset}
  A {\it (finite, measured) quantum set} is a pair $X = (\cO(X), \psi_X)$, where $\cO(X)$ is a finite dimensional C$^\ast$-algebra and $\psi_X:\cO(X) \to \bC$ is a faithful state.

  We write $|X|$ for $\dim\cO(X)$, and refer to this value as the {\it cardinality} or {\it size} of $X$. 
\end{definition}

The reason for our choice of notation is that when $\cO(X)$ is commutative, Gelfand theory tells us that we are really just talking about a finite set $X$ (the spectrum of $\cO(X)$) equipped with a probability measure $\mu_X$ defined  $\psi_X(f) = \int_X f(x)d\mu_X(x)$ for each $f \in \cO(X)$.   

Let $X = (\cO(X), \psi_X)$ be a quantum set.  Let $m_X:\cO(X) \otimes \cO(X) \to \cO(X)$ and $\eta_X:\bC \to \cO(X)$ be the multiplication and unit maps, respectively.  In what follows, we will generally only be interested in a special class of  finite quantum sets -- namely those that are measured by a {\it $\delta$-form} $\psi_X$, which we now define:  

\begin{definition}[\cite{Ba02}]
Let $\delta > 0$.  A state $\psi_X:\cO(X) \to \bC$ is called a {\it $\delta$-form} \cite{Ba02} if \[m_Xm_X^* = \delta^2 \iota,\] where the 
adjoint is taken with respect to the Hilbert space structure on $\cO(X)$ coming from the GNS construction with respect to $\psi_X$.
\end{definition}
For purposes of distinguishing between the Hilbert and C$^\ast$-structures on $\cO(X)$, we denote this Hilbert space by $L^2(X)$.

The most basic examples of $\delta$-forms are given by the uniform measure on the $n$-point set $X = [n]$ and the canonical normalized trace on $M_n(\bC)$.   In the first case, a simple calculation shows that $m^*(e_i) = ne_i \otimes e_i$, where $(e_i = e_i^* = e_i^2)_{i=1}^n$ is the standard basis of projections for $\cO(X)$, and so we have $\delta = \sqrt{n}$.  In the second case, one can show that $m^*(e_{ij}) = n \sum_{k=1}^n e_{ik} \otimes e_{kj}$, where $(e_{ij})_{1 \le i,j \le n}$ are the matrix units for $M_n(\bC)$.  So in this case we have $\delta = n$. More generally, if we have a multimatrix decomposition $\cO(X) = \bigoplus_{i=1}^sM_{n(i)}(\bC)$ and $\psi_X = \bigoplus_{i = 1}^s \text{Tr}(Q_i \cdot)$ is a faithful state (so $0 < Q_i \in M_{n(i)}(\bC)$ and $\sum_i \text{Tr}(Q_i) = 1$), then $\psi_X$ is a $\delta$-form if and only if $\text{Tr}(Q_i^{-1}) = \delta^2$ for all $1 \le i \le s$.  In particular, $\cO(X)$ admits a unique tracial $\delta$-form with $\delta^2 = \dim \cO(X)$ given by $\psi_X = \bigoplus_{i = 1}^s \frac{n(i)}{|X|}\text{Tr}(\cdot)$.

\begin{convention}
Unless otherwise stated, we assume from now on that the quantum sets we consider equipped with $\delta$-forms.
\end{convention}

We now endow quantum sets with an additional structure of a quantum adjacency matrix, turning then into quantum graphs.  The following definition of a quantum adjacency matrix/graph is a generalization of the \cite[Definition 5.1]{MuRuVe18} to our framework.

\begin{definition} \label{qgraph}
Let $X$ be a quantum set equipped with a $\delta$-form $\psi_X$.  A self-adjoint linear map $A_X:L^2(X) \to L^2(X)$ is called a {\it quantum adjacency matrix} if it has the following properties
\begin{enumerate}
\item \label{id} $m_X(A_X \otimes A_X)m_X^* = \delta^2A_X$.
\item \label{symmetric} $(\iota \otimes \eta_X^*m_X)(\iota \otimes A_X \otimes \iota)(m_X^*\eta_X\otimes \iota) = A_X$
\item \label{diag} $m_X(A_X \otimes \iota)m_X^* = \delta^2\iota$
\end{enumerate}
We call the triple $X = (\cO(X), \psi_X, A_X)$ a {\it quantum graph}.
\end{definition}

\begin{remark}  
In the special case where $\cO(X)$ is equipped its unique {\it tracial} $\delta$-form, then the definition of a quantum graph given here is equivalent to the one given in \cite{MuRuVe18}.  In addition, as explained in \cite{MuRuVe18}, a quantum graph $X = (\cO(X), \psi_X,A_X)$, where $\cO(X)$ is a commutative C$^\ast$-algebra, captures precisely the notion of a classical graph.   Indeed, in this case the spectrum $X$ of $\cO(X)$ is a finite set and $\psi_X$ is the uniform probability measure on $X$.  If we write $A_X$ as a matrix $A_X = [a_{ij}]_{i,j \in X}$ with respect to the canonical orthonormal basis of normalized projections $(\sqrt{n}e_i)_{i=1}^{n} \subset L^2(X)$, then conditions \eqref{id}, \eqref{symmetric} and \eqref{diag} say, respectively, that 
\[
a_{ij}^2 = a_{ij}, \quad a_{ij} = a_{ji}, \quad a_{ii}=1 \qquad (i,j \in X).
\]
In other words, $X$ is the vertex set of a classical graph (as defined in Section \ref{graphgame}) with adjacency matrix $A_X-I_n$.  Thus, in the quantum definition of a graph, we choose to work with {\it reflexive graphs} ($(v,v) \in E(X) \ \forall v \in V(X)$).  This choice is purely cosmetic from the perspective of (quantum) symmetries of graphs, in the sense that we have a bijection between (quantum) symmetries of reflexive graphs and those of their irreflexive versions.
\end{remark}

\begin{remark}
Note that any quantum set $X$ equipped with a $\delta$-form $\psi_X$ can be trivially upgraded to a quantum graph in two ways.  The first way is by declaring $A_X = \delta^2 \psi_X(\cdot)1$.  The second is by declaring $A_X = \iota$.   In the case of classical finite sets $X$, these constructions correspond to the complete graph $K_{|X|}$ and its (reflexive) complement $\overline{K_{|X|}}$, respectively.  For general quantum sets $X$ equipped with the quantum adjacency matrix $A_X = \delta^2 \psi_X(\cdot)1$, we will call these graphs {\it quantum complete graphs}.  For a general quantum graph $X$, we can also talk about its (reflexive) complement $\overline{X}$, which is given by $\overline{X} = (\cO(X), \psi_X, A_{\overline{X}})$ with $A_{\overline{X}} = \delta^2 \psi_X(\cdot)1 + \iota - A_X$.  With this definition we have that the complement of a quantum complete graph $X$ is the ``edgeless'' quantum graph $\overline{X} = (\cO(X), \psi_X, \iota)$. 
\end{remark}

We now introduce the quantum automorphism groups of quantum graphs.  The definition of these quantum automorphism groups follows along the same lines as for the quantum automorphism groups of classical graphs \cite{Ba05} and also the quantum automorphism groups of quantum sets \cite{Wan98, ban-sym, Ba02}.

\begin{definition} \label{def-qaut}
  Let $X = (\cO(X), \psi_X, A_X)$ be a quantum graph with $n = |X|$ and fix an orthonormal basis $\{e_i\}_{i=1}^n$ for $L^2(X)$.  Define $\cO(G_X)$ to be the universal unital $\ast$-algebra  generated by the coefficients $u_{ij}$ of a unitary matrix $u = [u_{ij}]_{i,j = 1}^n \in M_n(\cO(G_X))$ subject to the relations making  the map 
\[
\rho_X:\cO(X) \to \cO(X) \otimes \cO(G_X); \qquad \rho_X(e_i) = \sum_k e_j \otimes u_{ji}
\]
 a unital $\ast$-homomorphism satisfying the $A_X$-covariance condition  $\rho_X(A_X \cdot) = (A_X \otimes \iota)\rho_X$.
\end{definition}

The notation $\cO(G_X)$ is meant to convey the notion that the algebra consists of representative functions on a CQG $G_X$. Specifically, it is the ``largest'' CQG acting on $X$ so as to preserve the measure $\psi_X$ and graph structure $A_X$. This is formalized in the following result, whose proof is a straightforward application of the universality implicit in \Cref{def-qaut}.

\begin{proposition}
  The $\ast$-algebra $A=\cO(G_X)$ admits a Hopf $\ast$-algebra structure defined by
  \begin{equation*}
    \Delta u_{ij} = \sum_{k=1}^n u_{ik} \otimes u_{kj}, \quad S(u_{ij}) = u_{ji}^*, \quad \epsilon(u_{ij}) = \delta_{ij} \qquad (1 \le i,j \le n).
  \end{equation*}

  Furthermore, the action of $G_X$ on $X$ given by $\rho_X$ preserves $\psi_X$ in the sense that
\begin{equation*}
  (\psi_X\otimes\iota)\rho_X = \psi_X(\cdot)1:\cO(X)\to \cO(G_X). 
\end{equation*}
We call $G_X$ the {\it quantum automorphism group of the quantum graph $X$}.
\end{proposition}

\begin{proof}
This is a direct computation that we leave to the reader.  In fact a proof of this result will also follow as special case of the more general arguments presented following Remark \ref{bgvshopf}.  

%Checking, for instance, that the $n^2$ entries of the matrix
 % \begin{equation*}
 %   \left[\sum_{k=1}^n u_{ik} \otimes u_{kj}\right]_{i,j}\in M_n(A\otimes A)
  %\end{equation*}
  %satisfy the relations imposed on $[u_{ij}]_{i,j}$, similarly for
  %\begin{equation*}
    %[\delta_{ij}]_{i,j} = I_n\in M_n 
 % \end{equation*}
 % and so on.

\end{proof}

\begin{remark}
  Quantum automorphism groups are natural quantum analogues of their classical counterparts. Indeed, the abelianization of  $\cO(G_X)$ is exactly $\cO(\text{Aut}(X))$, the algebra of complex-valued functions on the group of automorphisms of the graph $X$. 
\end{remark}

\begin{example}
  When $X$ is a quantum complete graph, then $G_X$ is none other than Wang's quantum automorphism group of the finite space $(\cO(X), \psi_X)$ \cite{Wan98, Ba02}. In particular, the quantum automorphism group of the classical complete graph $K_n$ is precisely the quantum symmetric group $S^+_n$ of \Cref{ex.s}.
\end{example}

%%%%%%%%%%%%%%%%%%%%%%%%%%%%%%%%%%%%%%%%%%%%%%%%%%%%%%%%%%%%%%%%%%%%%%%%%%%%%
\subsection{Monoidal equivalence and bigalois extensions} \label{reps}

For a CQG $G$, we define the {\it representation category of $G$}, $\text{Rep}(G)$, to be the category whose objects are (equivalence classes of) finite dimensional representations of $G$, and whose morphisms are given by the intertwiner spaces $\{\text{Mor}(u,v)\}$.  The category $\text{Rep}(G)$ has a lot of nice structure, in particular it is an example of a so called  {\it strict C$^\ast$-tensor category with conjugates}.  See \cite{NeTu13} for more details.

We now come to a notion of central importance in this work: monoidal equivalence of compact quantum groups.
Let $G$ be a CQG.  Denote by $\text{Irr}(G)$ the set of equivalence classes of irreducible objects in $\text{Rep}(G)$.

\begin{definition}[\cite{Bi99, BiDeVa06}]  \label{monequiv}
Let $G_1, G_2$ be two compact quantum groups.  We say that $G_1$ and $G_2$ are {\it monoidally equivalent}, and write $G_1 \sim^{mon} G_2$, if there exists a bijection \[\varphi:\text{Irr}(G_1) \to \text{Irr}(G_2)\]  together with linear isomorphisms 
\[\varphi: \text{Mor}(u_1 \otimes \ldots \otimes u_n, v_1 \otimes \ldots\otimes v_m) \to  \text{Mor}(\varphi(u_1)\otimes\ldots \otimes \varphi(u_n), \varphi(v_1) \otimes \ldots \otimes \varphi(v_m)) \]
such that $\varphi(1_{G_1}) = 1_{G_2}$ ($1_{G_i}$ being the trivial representation of $G_i$), and such that for any morphisms $S,T$, 
\begin{align*}
\varphi(S \circ T) &=\varphi(S)\circ \varphi(T) \quad (\text{whenever $S \circ T$ is well-defined})\\
\varphi(S^*)&=\varphi(S)^* \\
\varphi(S \otimes T) &=\varphi(S) \otimes \varphi(T).
\end{align*}
\end{definition}
A monoidal equivalence between $G_1$ and $G_2$ means that the strict C$^\ast$-tensor categories $\text{Rep}(G_i)$ are {\it unitarily monoidally equivalent}.  More precisely, the maps $\varphi$ defined above canonically extend to a unitary tensor functor $\varphi: \text{Rep}(G_1) \to \text{Rep}(G_2)$ that is {\it fully faithful} (i.e., $\varphi$ defines an isomorphism between $\text{Mor}(u,v)$ and $\text{Mor}( \varphi(u), \varphi(v))$ for any objects $u,v \in \text{Rep}(G_1)$) and is {\it essentially surjective} (i.e., every object in $\text{Rep}(G_2)$ is of the form $\varphi(u)$ for some $u \in \text{Rep}(G_i)$).

% % Note that for general unitary tensor functors $\varphi:\text{Rep}(G_1) \to \text{Rep}(G_2)$, the equalities $\varphi(\otimes_i u_i) = \otimes_i\varphi(u_i)$ for $u_i \in \text{Rep}(G_1)$ do not necessarily hold, and we only require that $\varphi$ come with compatible natural unitary isomorphisms $E: \varphi(\otimes_i u_i) \to \otimes_i\varphi(u_i)$.  However in our case the fiber functors $\varphi$ arising from the bijection in Definition \ref{monequiv} always has this additional symmetry. See \cite[Chapter 2]{NeTu13} for more details.
% %

\subsubsection{Bigalois extensions}

 We now discuss an equivalent, but somewhat more concrete, way to think about monoidal equivalence of compact quantum groups.  The key object is that of a bigalois extension, which has its origins in Hopf algebra theory, but is adapted here to the analytic setting of CQGs.

Let $A = \cO(G)$ be a Hopf $\ast$-algebra of representative functions on a CQG $G$.   A {\it left $A$ $\ast$-comodule algebra} is a unital $\ast$-algebra $Z$ equipped with a unital $\ast$-homomorphism $\alpha:Z \to A \otimes Z$ satisfying $(\iota \otimes \alpha)\alpha = (\Delta \otimes \iota)\alpha$  and $(\epsilon \otimes \iota ) \alpha = \iota$.  Similarly, a  {\it right $A$ $\ast$-comodule algebra} is a unital $\ast$-algebra $Z$ equipped with a unital $\ast$-homomorphism $\beta:Z \to Z \otimes A$ satisfying $(\beta \otimes \iota)\beta = (\iota \otimes \Delta)\beta$  and $(\iota \otimes \epsilon ) \beta = \iota$. 

A left $A$ $\ast$-comodule algebra $(Z, \alpha)$ is called a {\it left $A$ Galois extension} if the linear map \[\kappa_l: Z \otimes Z \to A \otimes Z; \qquad \kappa_l(x \otimes y) = \alpha(x)(1 \otimes y)\]is bijective.  Similarly, a right $A$ $\ast$-comodule algebra $(Z, \beta)$ is called a {\it right $A$ Galois extension} if the linear map \[\kappa_r: Z \otimes Z \to Z \otimes A; \qquad \kappa_r(x \otimes y) = (x \otimes 1)\beta(y)\] is bijective.  Finally, let $A$ and $B$ be Hopf $\ast$-algebras.  A unital $\ast$-algebra $Z$ is called an {\it $A-B$ bigalois extension} if it is both a left $A$ Galois extension and a right $B$ Galois extension, and $Z$ is an $A-B$-bicomodule algebra.  I.e., if $\alpha, \beta$ denote the left and right comodule maps, respectively, then we have the equality of maps
\[
(\iota \otimes \beta)\alpha = (\alpha \otimes \iota)\beta: Z \to A \otimes Z \otimes B.
\]

\begin{remark}
  The notion of a (bi)galois extension should be regarded as a quantum analogue of the familiar notion of a (bi)torsor (or principle homogeneous (bi)bundle) in the context of group actions: If $G$ is a (finite) group and $G \curvearrowright X$ is an action of $G$ on a finite space $X$, we call $X$ a (left) {\it $G$-torsor} if the action is free and transitive.  This is equivalent to saying that the canonical map
\[
G \times X \to X \times X; \qquad (g,t) \mapsto (g\cdot t,t) 
\]  
is bijective.  Letting $\cO(X)$ denote the $\ast$-algebra of functions on $X$, then $\cO(X)$ is a left $\cO(G)$ $\ast$-comodule algebra with the map 
\[
\alpha:\cO(X) \to \cO(G) \otimes \cO(X)\cong\cO(G \times X); \qquad \alpha(f)(g,t) = f(g\cdot t).
\]
With these definitions, it is clear that $G \curvearrowright X$ is free and transitive if and only if 
\[
\kappa_l:\cO(X) \otimes \cO(X) \to \cO(G) \otimes \cO(X); \qquad \kappa_l(x \otimes y)(g,t) = \alpha(x)(1\otimes y)(g,t) = x(g\cdot t)y(t)
\]
is bijective, i.e., if and only if $\cO(X)$ is a left $\cO(G)$-galois extension.  Similar statements hold for right $G$-spaces and left-right $G_1$-$G_2$-spaces.
\end{remark}

In the following, we will be interested in bigalois extensions which admit non-zero C$^\ast$-envelopes.  The main way in which this is achieved is by considering necessary and sufficient conditions for the existence of invariant states on bigalois extensions.  In what follows, a {\it state} on a unital $\ast$-algebra $Z$ is a linear functional $\omega:Z \to \bC$ such that $\omega(1) = 1$ and $\omega(z^*z) \ge 0$ for all $z \in Z$. 

\begin{definition}
Let $Z$ be an $A-B$-bigalois extension.  A state $\omega:Z \to \bC$ is called {\it left-invariant} if $(\iota \otimes \omega)(z) =\omega (z)1_A$ for each $z \in Z$, and it is called right-invariant if $(\omega \otimes \iota)(z) =\omega (z)1_B$ for each $z \in Z$.  We call $\omega$ {\it bi-invariant} if it is both left and right-invariant.
\end{definition} 

\begin{example}
  The Hopf $\ast$-algebra $A = \cO(G)$ of representative functions on a compact quantum group $G$ is a natural example of an $A-A$-bigalois extension admitting a bi-invariant state.  Indeed, just take $\omega = h$, the Haar state on $A$.
\end{example} 

The following theorem summarizes some useful properties of bi-invariant states on bigalois extensions.  It is an amalgamation of various results in \cite{Bi99, BiDeVa06,De07}.

\begin{theorem} \label{invariant-states}
Let $G_1, G_2$ be compact quantum groups with $A = \cO(G_1)$ and $B = \cO(G_2)$.  Let $Z$ be an $A-B$-bigalois extension.  Then we have the following.
\begin{enumerate}
\item Any left/right/bi-invariant state $\omega:Z \to \bC$ is unique and faithful (if it exists).
\item The following are equivalent:
\begin{enumerate}
\item $Z$ admits a non-zero $\ast$-representation as bounded linear operators on a Hilbert space.
\item $Z$  admits a state.
\item $Z$ admits a bi-invariant state.
\item $Z$ admits a left (resp. right)-invariant state.
\end{enumerate}
\item If $Z$ admits a bi-invariant state $\omega$, denote by $B^u(G_1,G_2) \ne 0$ the universal C$^\ast$-algebra generated by $Z$ and by $B_r = \overline{\pi_\omega(Z)}$, the C$^\ast$-algebra generated by  the GNS representation with respect to $\omega$.  Then $\omega$ extends to a KMS state on both $B^u$ and $B_r$.  Moreover, $\omega$ is a tracial state if and only if both $G_1$ and $G_2$ are of Kac type.  (I.e., the Haar states on both $\cO(G_i)$ are tracial)
\end{enumerate} 
\end{theorem}

\begin{theorem}[\cite{Bi99, BiDeVa06}] \label{thm-mon-bg}
Let $G_1, G_2$ be compact quantum groups.  Then $G_1$ and $G_2$ are monoidally equivalent if and only if there exists an $\cO(G_1)-\cO(G_2)$-bigalois extension $Z$ equipped with a bi-invariant state $\omega$.
\end{theorem}

We refer the reader to \cite[Theorem 2.3.11]{NeTu13} or \cite[Theorem 3.9 and Proposition 3.13]{BiDeVa06} for a precise description of the the bigalois extension $(Z,\omega)$ induced by the monoidal equivalence in Theorem \ref{thm-mon-bg}.

We end this section by stating a simple criterion due to Bichon \cite{Bi99} (for compact matrix quantum groups) for a bigalois extension to admit an invariant state $\omega$.  First we need some definitions.  Let $n \in \bN$ and $F_i \in \text{GL}_n(\bC)$.  We define $\cO(U^+_{F_1}, U^+_{F_2})$ to be the unital $\ast$-algebra generated by the coefficients $z_{ij}$ of a $n_1 \times n_2$ matrix $z = [z_{ij}]_{\substack {1 \le i \le n_1 \\ 1 \le j \le n_2}} \in M_{n_1, n_2}(\cO(U^+_{F_1}, U^+_{F_2}))$  satisfying the relations making both $z$ and $F_1\bar z F_{2}^{-1}$ unitary, where $\bar z = [z_{ij}^*]$.  When $F_1 = F_2 = F$, note that $\cO(U^+_{F}, U^+_{F})= \cO(U_F^+)$ is the Hopf $\ast$-algebra of representative functions on the universal unitary quantum group $U^+_F$ introduced earlier.  We also note that if $\cO(U^+_{F_1}, U^+_{F_2}) \ne 0$ then $\cO(U^+_{F_1}, U^+_{F_2}) $ is an $\cO(U^+_{F_1})-\cO(U^+_{F_2})$-bigalois extension with respect to the bicomodule structure given by 

\begin{align*}
&\alpha_{F_1,F_2}:\cO(U^+_{F_1}, U^+_{F_2}) \to \cO(U^+_{F_1}) \otimes \cO(U^+_{F_1}, U^+_{F_2}); \qquad \alpha_{F_1,F_2}(z_{ij}) = \sum_{k=1}^{n_1}u_{ik} \otimes z_{kj} \\
&\beta_{F_1, F_2}: \cO(U^+_{F_1}, U^+_{F_2})\to \cO(U^+_{F_1}, U^+_{F_2}) \otimes \cO(U_{F_2}^+); \qquad \beta_{F_1,F_2}(z_{ij}) = \sum_{l=1}^{n_2} z_{il} \otimes v_{lj}, 
\end{align*}
where $u = [u_{ij}], v = [v_{ij}]$ are the fundamental representations of $U_{F_1}^+, U^+_{F_2}$, respectively.

\begin{theorem}[Proposition 6.2.6 in \cite{Bi99}]\label{sufficient-state}
Let $G$ be a compact matrix quantum group and $(Z, \alpha)$ a left $\cO(G)$-Galois extension.  Let $F \in \text{GL}_n(\bC)$ be such that $G < U^+_F$ (with corresponding surjective morphism $\pi:\cO(U^+_F) \to \cO(G)$).  If there exists $F_1 \in \text{GL}_{n_1}(\bC)$ and a surjective $\ast$-homomorphism $\sigma:\cO(U^+_{F}, U^+_{F_1}) \to Z$ satisfying $\alpha \circ \sigma = (\pi \otimes \sigma)\alpha_{F,F_1}$, then $Z$ admits a left-invariant state $\omega:Z \to \bC$.
\end{theorem}

%%%%%%%%%%%%%%%%%%%%%%%%%%%%%%%%%%%%%%%%%%%%%%%%%%%%%%%%%%%%%%%%%%%%%%%%%%%%%
%%%%%%%%%%%%%%%%%%%%%%%%%%%%%%%%%%%%%%%%%%%%%%%%%%%%%%%%%%%%%%%%%%%%%%%%%%%%%
\section{Quantum isomorphisms of graphs and bigalois extensions}\label{se.bigal}

The aim of this section is to show that a quantum isomorphism between two graphs $X$ and $Y$ is nothing other than a (quotient of a) $\cO(G_Y)$-$\cO(G_X)$-bigalois extension in disguise.  We begin by extending the definition of the graph isomorphism game $\ast$-algebra $\cA(Iso(X,Y))$ to include quantum graphs.

\begin{definition} \label{def-linking}
  Let $X=(\cO(X),\psi_X,A_X)$ and $Y=(\cO(Y),\psi_Y,A_Y)$ be quantum graphs with $|X|=n$ and $|Y|=m$, and fix orthonormal bases $\{e_j\}$ and $\{f_i\}$ for $\cO(X)$ and $\cO(Y)$ relative to $\psi_X$ and $\psi_Y$ respectively.  Let $\cO (G_Y, G_X)$ be the universal $\ast$-algebra generated by the entries $p_{ij}$ of a unitary matrix
  \begin{equation*}
    p = [p_{ij}]_{ij}\in \cO(G_Y,G_X)\otimes B(L^2(X), L^2(Y))
  \end{equation*}
  with relations ensuring that
  \begin{equation*}
    \rho_{Y,X}:\cO(X) \to \cO(Y)\otimes \cO(G_Y,G_X); \qquad e_j\mapsto \sum_{i} f_i\otimes p_{ij} 
  \end{equation*}
  is a unital $\ast$-homomorphism satisfying
  \begin{equation}\label{eq:7}
    \rho_{Y,X}(A_X\cdot) = (A_Y\otimes\iota)\rho_{Y,X}. 
  \end{equation}
\end{definition}
Our first observation is that the above morphism $\rho_{Y,X}$, if it exists, is automatically state-preserving.

\begin{lemma}
Assume $\cO(G_Y,G_X) \ne  0$.  Then the morphism $ \rho_{Y,X}:\cO(X) \to \cO(Y)\otimes \cO(G_Y,G_X)$ is state-preserving in the sense that
\begin{equation}\label{eq:1}
      (\psi_Y\otimes\iota)\rho_{Y,X} = \psi_X(\cdot)1:\cO(X)\to \cO(G_Y,G_X).
  \end{equation}
\end{lemma}

\begin{proof}
Consider the matrix $p = [p_{ij}]_{ij}\in \cO(G_Y,G_X)\otimes B(L^2(X), L^2(Y))$, viewed canonically as a linear map \[p:L^2(X) \otimes \cO(G_Y,G_X) \to L^2(Y) \otimes \cO(G_Y,G_X); \qquad p(\xi \otimes a) = \sum_{i,j} |f_i\rangle \langle e_{i}|\xi\rangle \otimes p_{ij}a. \]
It then follows that $\rho_{Y,X}(\xi) = p(\xi \otimes 1)$ for each $\xi \in L^2(X)$ (here and below we are identifying $L^2(X)$ and  $L^2(Y)$ with $\cO(X)$ and $\cO(Y)$).   Consider now the $\cO(G_Y, G_X)$-valued sesquilinear forms on $L^2(X)\otimes\cO(G_Y,G_X)$ and $L^2(Y)\otimes\cO(G_Y,G_X)$ given by 
\[
\langle \xi_1 \otimes a |\xi_2 \otimes b \rangle_{L^2(X)\otimes\cO(G_Y,G_X)} = b^*a \psi_X(\xi_2^*\xi_1) \quad \& \quad \langle \eta_1 \otimes a |\eta_2 \otimes b \rangle_{L^2(Y)\otimes\cO(G_Y,G_X)} = b^*a \psi_Y(\eta_2^*\eta_1).  
\]
Then a simple calculation using the fact that $p^*p = 1$  and $p(1\otimes 1) = \rho_{Y,X}(1) = 1 \otimes 1$ gives 
\begin{align*}
(\psi_Y \otimes \iota) \rho_{Y,X}(\xi) &=  (\psi_Y \otimes \iota) p(\xi \otimes 1) \\
&=\langle p(\xi \otimes 1) |1 \otimes 1 \rangle_{L^2(Y)\otimes\cO(G_Y,G_X)} \\
&=\langle p(\xi \otimes 1) |p(1 \otimes 1) \rangle_{L^2(Y)\otimes\cO(G_Y,G_X)} \\
&= \langle p^*p(\xi \otimes 1) |1 \otimes 1 \rangle_{L^2(X)\otimes\cO(G_Y,G_X)}\\
&=\langle(\xi \otimes 1) |1 \otimes 1 \rangle_{L^2(X)\otimes\cO(G_Y,G_X)}\\ 
&= \psi_X(\xi)1.
\end{align*}

\end{proof}

\begin{remark} \label{bgvshopf} When the two quantum graphs coincide we have $\cO(G_X, G_X) = \cO(G_X)$ (the Hopf $\ast$-algebra of polynomial functions on the quantum automorphism group $G_X$) and $\rho_{X,X} = \rho_X$.  For classical graphs $X, Y$, we have $\cO(G_Y, G_X) = \cA(Iso(Y,X))$.  Indeed, the fact that $\rho_{Y,X}$ is a unital $\ast$-homomorphism intertwining the quantum adjacency matrices $A_X$ and $A_Y$ says exactly that the unitary matrix $p = [p_{ij}]$ satisfies $(1 \otimes A_Y)p = p(1 \otimes A_X)$ and has entries which are self-adjoint projections satisfying $\sum_i p_{ji} = 1 = \sum_j p_{ji}$, $p_{ji}p_{jl} = \delta_{il}p_{ji}$, and $p_{ij}p_{lj} = \delta_{i}p_{ij}$.  Compare with Proposition \ref{Iso-alg} and Remark \ref{rem-iso}. See also \cite{amrssv, Lu17}.
\end{remark}

With the above in mind, we now provide a natural extension of the notion of quantum isomorphism to our quantum graphs.  Compare with \cite{MuRuVe18}.

\begin{definition}
Let $X,Y$ be quantum graphs.  We say that $X$ is {\it algebraically quantum isomorphic}
to $Y$ if $\cO(G_Y,G_X) \ne 0$, and write $X \cong_{A^*}Y$.  If $\cO(G_Y,G_X) $ admits a non-zero C$^\ast$-representation, then we say that  $X$ is {\it C$^\ast$-algebraically quantum isomorphic}
to $Y$, and write $X \cong_{A^*}Y$.  Finally, we say $X \cong_{qc}Y$ if $\cO(G_Y,G_X)$ admits a tracial state (following the existent notation for classical graphs).
\end{definition}

For the remainder of the present discussion we fix two quantum graphs $X, Y$ as above and assume that the $\ast$-algebra $\cO (G_Y, G_X)$ is non-zero.  Our aim is to show that $\cO (G_Y, G_X)$ admits a natural structure as a $\cO(G_Y)$-$\cO(G_X)$ bigalois extension.  

Consider the comodule-algebra structure map
\begin{equation*}
  \rho_{Y}:\cO(Y)\to \cO(Y)\otimes \cO(G_Y). 
\end{equation*}
By the universality of $\rho_{Y,X}$, the composition 
\begin{equation*}
  (\rho_Y\otimes\mathrm{id})\circ\rho_{Y,X}:\cO(X)\to \cO(Y)\otimes \cO(G_Y)\otimes \cO(G_Y,G_X)
\end{equation*}
must factor as $(\id \otimes\alpha)\circ \rho_{Y,X}$ for a unique $*$-algebra morphism
\begin{equation*}
  \alpha:\cO(G_Y,G_X)\to \cO(G_Y)\otimes \cO(G_Y,G_X)
\end{equation*}
given simply by
\begin{equation*}
\alpha(p_{ij}) = \sum_k u_{ik} \otimes p_{kj}, 
\end{equation*}
where $u = [u_{ij}]$ is the fundamental representation of $\cO(G_Y)$.

Similarly, $\cO(G_Y,G_X)$ has a right $\cO(G_X)$ $\ast$-comodule algebra structure given by
\[\beta: \cO(G_Y,G_X) \rightarrow \cO(G_Y,G_X) \otimes \cO(G_X); \qquad 
\beta(p_{ij}) = \sum_k p_{ik} \otimes v_{kj},
\] 
where $v = [v_{ij}]$ is the fundamental representation of $\cO(G_X)$.  It is also clear that $\cO (G_Y, G_X)$ is an $\cO (G_Y)-\cO(G_X)$ bicomodule with respect to $\alpha$ and $\beta$.

Continuing in the same vein, we can define ``cocomposition'' $\ast$-morphisms
\begin{align*}
\gamma_Y: \cO(G_Y) \to \cO(G_Y,G_X) \otimes \cO(G_X,G_Y); \qquad \gamma_Y(u_{ij}) = \sum_k p_{ik} \otimes q_{ki}\\
\gamma_X : \cO(G_X) \to \cO(G_X,G_Y) \otimes \cO(G_Y,G_X); \qquad \gamma_X(v_{ij}) = \sum_k q_{ik} \otimes p_{kj} \\
\end{align*}
where $q = [q_{ij}]$ is the matrix of generators of $\cO(G_X,G_Y)$.  For example, to construct $\gamma_Y$, we consider the morphism
\[
(\rho_{Y,X} \otimes \iota)\rho_{X,Y}:\cO(Y) \to \cO(Y) \otimes \cO(G_Y,G_X) \otimes \cO(G_X,G_Y).
\]
By universality of $\rho_Y$, there exists a unique morphism $\gamma_Y: \cO(G_Y) \to \cO(G_Y,G_X) \otimes \cO(G_X,G_Y)$ so that
\[
(\rho_{Y,X} \otimes \iota)\rho_{X,Y} = (\iota \otimes \gamma_Y)\rho_{X,Y}.\]
This map is readily seen to be given by the proposed formula above.

Thus far, the algebras $\cO(G_X)$, $\cO(G_Y)$, $\cO(G_Y,G_X)$ and $\cO(G_X,G_Y)$ together with the maps $\alpha$ and $\beta$, their analogues for $\cO(G_X,G_Y)$, and $\gamma_X$, $\gamma_Y$ constitute a two-object {\it cocategory} $\cC$ in the sense of \cite[Definition 2.1]{bch-cgr}: the four algebras are to be thought of as dual to ``spaces of morphisms'' between two objects ($x\to x$ for $\cO(G_X)$, $x\to y$ for $\cO(G_Y,G_X)$, etc.), and the $\gamma$ maps are dual to morphism composition. 

Next, we make $\cC$ into a {\it cogroupoid} in the sense of \cite[Definitions 2.3 and 2.4]{bch-cgr}: this entails defining ``coinversion'' maps
\begin{align}
  S_{X,Y}:\cO(G_X,G_Y)&\to \cO(G_Y,G_X)\label{eq:5}\\
  S_{Y,X}:\cO(G_Y,G_X)&\to \cO(G_X,G_Y)\label{eq:6},
\end{align}
which will require some preparation.

Let $F=F_X\in M_n$ and $G=F_Y\in M_m$ be matrices with the property that $Fe_i=e_i^*$ and similarly for $G$, so that $\overline{F}=F^{-1}$ and $\overline{G}=G^{-1}$. Note that the involutivity of the morphisms
\begin{align*}
  \rho_X:\cO(X) &\to \cO(X)\otimes \cO(G_X)\\
  \rho_Y:\cO(Y) &\to \cO(Y)\otimes \cO(G_Y)\\
  \rho_{Y,X}: \cO(X) &\to \cO(Y) \otimes \cO(G_Y, G_X)
\end{align*}
is equivalent, respectively, to the equalities
\begin{align*}
  (1 \otimes F)\bar{u} &= u(1 \otimes F)\\
  (1 \otimes G)\bar v &= v (1 \otimes G)\numberthis\label{eq:uvp}\\
  (1 \otimes G)\bar p &= p(1 \otimes F)
\end{align*}
We will henceforth abuse notation and write $uF$ for $u(1\otimes F)$, etc. Taking this into account, we have
\begin{equation*}
  G^{-1}pF= \overline{p} \text{ and similarly } F^{-1}qG = \overline{q}. 
\end{equation*}

It is now a simple check to see that
\begin{equation}\label{eq:3}
  f_i\mapsto \sum_j e_j \otimes p^*_{ij}
\end{equation}
defines a unital algebra homomorphism
\begin{equation}\label{eq:4}
  \cO(X)\to \cO(Y)\otimes \cO(G_Y,G_X)^{op}. 
\end{equation}
Applying $G$ to both sides of \Cref{eq:3}, writing $e_j=FF^{-1}e_j$ and using
\begin{equation*}
  Fe_j=e^*_j,\quad Gf_i=f^*_i,
\end{equation*}
it follows that \Cref{eq:3} is involutive with respect to the modified $*$-structure $\star$ on $\cO(G_Y,G_X)^{op}$ given by
\begin{equation*}
  (p^*)^{\star} = (F^{-1}p^* G)^t  
\end{equation*}
(the `$t$' superscript denoting the transpose). The defining universality property of $\cO(G_X,G_Y)$ then implies that the morphism \Cref{eq:4} given by \Cref{eq:3} factors as
\begin{equation*}
  (\iota\otimes S_{X,Y})\rho_{X}
\end{equation*}
for a conjugate-linear anti-morphism \Cref{eq:5}. $S_{Y,X}$ is defined similarly, and in summary we have 
\begin{align*}
  S_{X,Y}:\quad q&\mapsto p^*,\quad q^*\mapsto G^t \overline{p} F^{-t} \\
  S_{Y,X}:\quad p&\mapsto q^*,\quad p^*\mapsto F^t \overline{p} G^{-t}
\end{align*}
where the `$t$' superscript means `transpose' while `$-t$' denotes `transpose inverse'.

The morphisms \Cref{eq:5,eq:6} enrich the above-mentioned cocategory $\cC$ to a connected {\it cogroupoid} in the sense of \cite[Definitions 2.3 and 2.4]{bch-cgr}.

We are now ready for the main result of this section.

% % \begin{equation*}
% %   \psi_X(F e_i (F e_s)^*)=\delta_{is},\quad   \psi_Y(G f_j (G f_t)^*)=\delta_{jt}.
% % \end{equation*}
% % 

%(such that $*\circ S_{X,Y}$ and $*\circ S_{Y,X}$ are mutually inverse). 

% % To begin, denote by $u = [u_{ij}] \in M_n(\cO(G_Y))$ and $v = [v_{ij}] \in M_n(\cO(G_X))$ the fundamental representations of the corresponding quantum automorphism groups.  Then $\cO(G_Y,G_X)$ becomes a left $\cO(G_Y)$ $\ast$-comodule algebra via the map
% % \[
% % \alpha: \cO(G_Y, G_X) \to \cO(G_Y) \otimes \cO(G_Y,G_X); \qquad \alpha(p_{ij}) = \sum_k u_{ik} \otimes p_{kj}.
% % \]
% % One can readily check that the matrix $(\alpha \otimes \iota)(P) := [\alpha(p_{ij})]$ and its entries  satisfy the conditions of Definition \ref{def-linking}, so $\alpha$ is a well defined $\ast$-homomorphism.  The relations $(\iota \otimes \alpha)\alpha = (\Delta \otimes \iota)\alpha$  and $(\epsilon \otimes \iota ) \alpha = \iota$ are also easy to check. 
% %

\begin{theorem} \label{non-zero} 
If $\cO (G_Y, G_X)$ is non-zero, then $(\cO(G_Y, G_X), \alpha, \beta)$ is a $\cO(G_Y)$-$\cO(G_X)$-bigalois extension.
\end{theorem}
\begin{proof}
By \cite[Proposition 2.8]{bch-cgr} this is an immediate consequence of $\cC$ being a connected cogroupoid.  More precisely, the arguments therein show that the relevant linear maps     \begin{align*}
 &\kappa_l : \cO(G_Y, G_X) \otimes \cO(G_Y,G_X) \to \cO(G_Y) \otimes \cO(G_Y, G_X); \qquad \kappa_l(x \otimes y) = \alpha(x)(1 \otimes y) \\
 &\kappa_r: \cO(G_Y, G_X) \otimes \cO(G_Y,G_X)\to  \cO(G_Y,G_X) \otimes \cO(G_X); \qquad \kappa_r(x \otimes y) = (x \otimes 1)\beta(y) 
 \end{align*}
are bijective with explicit inverses given by 
 \begin{align*}
\eta_l : \cO(G_Y) \otimes \cO(G_Y,G_X) \rightarrow \cO(G_Y,G_X) \otimes \cO(G_Y,G_X); \qquad  \eta_l = (\iota \otimes m)(\iota \otimes S_{X,Y} \otimes \iota)(\gamma_Y \otimes \iota) \\
\eta_r : \cO(G_Y,G_X) \otimes \cO(G_X) \rightarrow \cO(G_Y,G_X) \otimes \cO(G_Y,G_X); \qquad  \eta_r = (m \otimes \iota)(\iota \otimes S_{X,Y} \otimes \iota)(\iota \otimes \gamma_X)
\end{align*}
where $m$ denotes the multiplication map in the appropriate algebra.
\end{proof}

\Cref{non-zero} puts some of the material in \cite{Lu17} in a category-theoretic perspective. To make sense of this, we need to recall

\begin{definition}\label{def.orb}
  The {\it quantum orbital algebra} of a (quantum) graph $X$ is the endomorphism algebra of $\cO(X)$ as a comodule over $\cO(G_X)$. That is, the algebra of intertwiners $\text{Mor}(u,u) \subset B(L^2(X))$, where $u$ denotes the fundamental representation of $G_X$.
\end{definition}

In the case of classical graphs, this is not quite \cite[Definition 3.10]{Lu17}, but is equivalent to it by \cite[Theorem 3.11]{Lu17}.  Note that \cite[Theorem 4.2]{Lu17} follows from \Cref{non-zero}: the former says that a quantum isomorphism between two (classical) graphs entails an isomorphism between their quantum orbital algebras that identifies the respective adjacency matrices. Since by \Cref{non-zero} we have a category equivalence
\begin{equation*}
  \text{Rep}(G_X) \simeq \text{Rep}(G_Y)
\end{equation*}
identifying $\cO(X)$ on the left to $\cO(Y)$ on the right, this implements an isomorphism between the endomorphism algebras of these two objects in the respective categories (i.e. the quantum orbital algebras). Furthermore, the fact that this isomorphism identifies $A_X$ and $A_Y$ follows from \Cref{eq:7}.

\subsection{Existence of states on $\cO(G_Y,G_X)$}\label{subse.states}
Our next result shows that $\cO(G_Y,G_X)$ always admits a faithful bi-invariant state (and hence a C$^\ast$-completion) whenever this algebra is non-zero. 

\begin{theorem} \label{traces-exist}
Let $X, Y$ be quantum graphs.  If $\cO(G_Y,G_X) \ne 0$, then there exists a faithful bi-invariant state $\omega: \cO(G_Y,G_X) \to \bC$, and therefore we have a monoidal equivalence of compact quantum groups $G_X \sim^{\text{mon}}G_Y$.  Moreover, $\omega$ is tracial if and only if both $G_X$ and $G_Y$ are of Kac type.  
\end{theorem}

\begin{proof}
Recall the matrices $F=F_X$ and $G=F_Y$ from the preceding discussion. The equations \Cref{eq:uvp} imply that we have surjective $\ast$-homomorphisms $\pi:\cO(U^+_{F_Y}) \to \cO(G_Y)$ and $\sigma: \cO(U^+_{F_Y}, U^+_{F_X}) \to \cO(G_Y, G_X)$ satisfying  $\alpha \circ \sigma = (\pi \otimes \sigma)\alpha_{F_Y,F_X}$.  By \Cref{sufficient-state,invariant-states} $\cO(G_Y, G_X)$ then admits a $\cO(G_Y)$-$\cO(G_X)$-invariant state (which is tracial precisely when $G_Y,G_X$ are both of Kac type).  By \Cref{thm-mon-bg}, $G_X \sim^{\text{mon}}G_Y$.
\end{proof}

\begin{corollary} \label{th:qiso-monoidal}
Let $X$ and $Y$ be quantum graphs.  Then the following are equivalent.
\begin{enumerate}
\renewcommand{\labelenumi}{(\arabic{enumi})}
\item $X \cong_{A^*}Y$.
\item $X \cong_{C^*}Y$.
\end{enumerate}
Moreover, if both $X$ and $Y$ are equipped with tracial $\delta$-forms, then $X \cong_{qc}Y$.
\end{corollary}
\begin{proof}
  $(2) \implies (1)$ by definition, while the converse follows from \Cref{traces-exist}. The same theorem also shows that when $G_X$ and $G_Y$ are Kac (as is the case if $X$ and $Y$ are equipped with tracial $\delta$-forms) $\cO(G_Y,G_X)$ is equipped with a trace. This proves the last claim.
\end{proof}

Restricting our attention to classical graphs $X$, and $Y$ we arrive at one of the main results of the paper.

\begin{theorem} \label{th.cls-grph}
Let $X$ and $Y$ be classical graphs.  Then the following conditions are equivalent.
\begin{enumerate}
\item \label{1} $X \cong_{A^*} Y$.
\item \label{2} $X \cong_{qc} Y$.
\item \label{3} $X \cong_{C^*} Y$.
\end{enumerate}

\end{theorem}

\begin{proof}
  This is an immediate consequence of \Cref{th:qiso-monoidal}.
\end{proof}

\begin{remark}
The above theorems show that the algebras $\cO(G_Y,G_X)$ are non-zero in the category of $\ast$-algebras if and only if $\cO(G_Y,G_X)$ admits a non-zero representation as bounded operators on Hilbert space.  In other words, the $\ast$-algebra and C$^\ast$-algebra worlds coincide for this class of examples.

One illustration of the distinction between $*$-algebras and C$^*$-algebras is in the behavior of projections (i.e. self-adjoint idempotents). In a C$^\ast$-algebra, if one has self-adjoint idempotents $\{p_1,...,p_N \}$ satisfying $p_1+ \cdots + p_N =1,$ then necessarily $p_i p_j=0, \, \forall i \ne j$.

The situation is very different for $\ast$-algebras, however. While triples of projections with sum $1$ still commute, quadruples need not. This can be seen, for instance, from \cite[\S 2.1]{bg-diamond}. There, the ring generated by three idempotents $a$, $b$ and $c$ whose sum is also idempotent ($d=1-(a+b+c)$ thus being idempotent as well) is shown to have a basis as a free abelian group consisting of those monomials in $a$, $b$ and $c$ such that
\begin{itemize}
\item no letter appears twice in succession;
\item $b$ never appears to the left of $a$.
\end{itemize}
This makes it clear that $ab\ne 0$. One can simply reprise this example over $\bC$ (i.e. working with complex algebras rather than rings) and superimpose a $*$-structure by requiring that $a$, $b$ and $c$ be self-adjoint. The result is a complex $*$-algebra with four non-orthogonal projections adding up to $1$. 

In fact, even more pathological examples exist. In \cite{HMPS17} a machine-assisted proof is given that the $\ast$-algebra $\cl A(Hom(K_5, K_4))$ is non-trivial.  This is a $\ast$-algebra with generators
\begin{equation*}
 \{ e_{x,a}: 1 \le x \le 5, 1 \le a \le 4 \} 
\end{equation*}
satisfying the usual relations, $e_{x,a}^* = e_{x,a}^2 = e_{x,a}$, $e_{x,a} e_{x,b} =0,$ when $a \ne b$, $\sum_{a=1}^4 e_{x,a} =1, \, \forall x$, and the relations, $e_{x,a}e_{y,a}=0, \, x \ne y$, prescribed by the graphs.

If one sets $p_a = \sum_x e_{x,a}$, then $p_a^2 = p_a = p_a^*$, for $1 \le a \le 4$. Hence, $q_a = 1 - p_a$ are also self-adjoint idempotents.  However,
\begin{equation*}
  \sum_{a=1}^4 q_a = 4 \cdot 1 - \sum_{a=1}^4 p_a = 4 \cdot 1 - \sum_{x=1}^5 \sum_{a=1}^4 e_{x,a} = 4 \cdot 1 - 5 \cdot 1 = -1.  
\end{equation*}
Thus, it is possible to have 4 self-adjoint idempotents sum to $-1$ in a $\ast$-algebra.
\end{remark}

\subsection{From monoidal equivalence to quantum isomorphism}
\Cref{traces-exist,th:qiso-monoidal} show that for a pair of quantum graphs $X,Y$ the condition $X \cong_{A^*} Y$ implies that the corresponding quantum automorphism groups $G_X$ and $G_Y$ are monoidally equivalent. Based on this connection between quantum isomorphism and monoidal equivalence, it is natural to ask whether the converse holds, namely: {\it Does $G_X \sim^{\text{mon}}G_Y \implies X \cong_{A^*} Y$?} 

The answer to this question turns out to be `no' in general. For example, take $X = K_n$ and $Y = \overline{K_n}$.  In this case we have $G_X = G_Y = S_n^+$ (so $G_X$ and $G_Y$ are in particular monoidally equivalent as compact quantum groups), but it is clear from the definitions that $\cA(Iso(X,Y)) = \cO(G_X, G_Y) = 0$.  The intuitive reason for this is that the trivial monoidal equivalence taking $\text{Rep}(S_n^+)$ to itself does not map the adjacency matrix $A_X$ to $A_Y$.  In fact, {\it there cannot exist any} any monoidal equivalence $\varphi: \text{Rep}(S_n^+) \to \text{Rep}(S_n^+) $ satisfying $\varphi(A_X) = A_Y$.  This is because such a monoidal equivalence would force $A_X$ and $\varphi(A_X) = A_Y$ to be isospectral.

On the other hand, the following theorem shows that whenever we have a quantum group $G$ monoidally equivalent to $G_X$, it is possible to find a quantum graph $Y$ so that $G = G_Y$ and $X \cong_{A^*} Y$.

% % We will see that the above observations in the context of $K_n$ and $\overline{K_n}$ turn out to be the only real obstruction to a monoidal equivalence $G_X \sim^{\text{mon}}G_Y$ implying a quantum  isomorphism $X \cong_A^* Y$ for general graphs $X,Y$.  The only caveat is that either $X$ or $Y$ may no longer be classical. 
% % 

% % We now turn to the problem of identifying (up to isomorphism), which compact quantum groups $G$ can be monoidally equivalent to the quantum automorphism group $G_X$ of a quantum graph $X$.  Our first crucial result here is that such a $G$ must also be a quantum automorphism group of a quantum graph $Y$.
% %

\begin{theorem} \label{stability-monoidal}
Let $X = (\cO(X),\psi_X, A_X)$ be a quantum graph and $G_X$ its quantum automorphism group.  Let $G$ be another compact quantum group that is monoidally equivalent to $G_X$.  Then there exists a quantum graph $Y = (\cO(Y), \psi_Y,A_Y)$ so that $G = G_Y$, and we have a quantum isomorphism $X \cong_{A^*}Y$.
\end{theorem}

\begin{proof}
When $X$ is a quantum complete graph, this result is already known \cite[Theorem 3.6.5]{De07}.  The proof in the case of arbitrary $X$ follows almost verbatim, so we just sketch the main ideas.  

Let $\varphi:\text{Rep}(G_X) \to \text{Rep}(G)$ be the unitary fiber functor implementing the monoidal equivalence as in Defintion \ref{monequiv}.  Put $L^2(Y) = \varphi(L^2(X))$, $d_Y = \dim(L^2(Y))$ and let $v = \varphi(u) \in M_{d_Y}(\cO(G))$ be the corresponding unitary representation of $G$ on $L^2(Y)$.  Put $m_Y = \varphi(m_X) \in \text{Mor}(v \otimes v,v)$, $\eta_Y = \varphi(\eta_X)\in \text{Mor}(1,v)$ and $\psi_Y = \eta_Y^* \in \text{Mor}(v,1)$ and $A_Y = \varphi(A_X) \in \text{Mor}(v,v)$.   Then exactly as in the proof of \cite[Theorem 3.6.5]{De07}, $L^2(Y)$ is a unital C$^\ast$-algebra with multiplication $m_Y$, unit $\eta_Y$,  involution $\sharp: \xi \mapsto \xi^\sharp = (\iota \otimes \xi^*)(m^*_Y \eta_Y)$, and $\psi_Y:L^2(Y) \to \bC$ is a $\delta$-form.  We denote this C$^\ast$-algebra by $\cO(Y)$.  Finally, consider the map $A_Y:L^2(Y) \to L^2(Y)$.  Then by definition of $\varphi$, we have 
\begin{align*}
&A_Y^* = \varphi(A_X)^* = \varphi(A_X^*) = \varphi(A_X) = A_Y, \\
& m_Y(A_Y \otimes A_Y)m_Y^* = \varphi(m_X(A_X \otimes A_X)m_X^*) =\varphi(\delta^2 A_X) =  \delta^2A_Y,\\
&(\iota \otimes \eta_Y^*m_Y)(\iota \otimes A_Y \otimes \iota)(m_Y^*\eta_Y\otimes \iota) = \varphi((\iota \otimes \eta_X^*m_X)(\iota \otimes A_X \otimes \iota)(m_X^*\eta_X\otimes \iota)) = \varphi(A_X) = A_Y,\\
&m_Y(A_Y \otimes \iota)m_Y^* = \varphi(m_X(A_X \otimes \iota)m_X^*) = \varphi (\delta^2\iota) = \delta^2 \iota,
\end{align*}
so $A_Y$ is a quantum adjacency matrix and $Y = (\cO(Y), \psi_Y, A_Y)$ is a quantum graph.

Now let $G_Y$ be the quantum automorphism group of $Y$, with fundamental representation $w \in M_{d_Y}(\cO(G_Y))$.  Then by Definition \ref{qgraph} and the construction of the morphisms $m_Y,\eta_Y, \varphi_Y, A_Y$ using the monoidal equivalence $\varphi$, there is a surjective Hopf $\ast$-homomorphism $\sigma:\cO(G_Y) \to \cO(G)$ given by $(\sigma \otimes \iota)w =v$.  In particular, $G < G_Y$ is a quantum subgroup, which implies that for any $m,n \in \bN_0$, we have $\text{Mor}(w^{\otimes m}, w^{\otimes n}) \subseteq \text{Mor}(v^{\otimes m}, v^{\otimes n})$.   To prove that in fact $G = G_Y$, it suffices to check equality in the above containments for each $m,n$ (see for example \cite[Proposition 3.5]{bcv}).  To this end, recall that by our monoidal equivalence, we have isomorphisms $\varphi: \text{Mor}(u^{\otimes m}, u^{\otimes n}) \cong \text{Mor}(v^{\otimes m}, v^{\otimes n})$.  Moreover, since (by universality of $G_X$) the space $ \text{Mor}(u^{\otimes m}, u^{\otimes n})$ is generated (in the C$^\ast$-tensor categorical sense) by the maps $\{\iota, m_X,  \eta_X, A_X\}$, it  follows that $\text{Mor}(v^{\otimes m}, v^{\otimes n})$ is also generated the images $\{\varphi(\iota), \varphi(m_X),  \varphi(\eta_X), \varphi(A_X)\} = \{\iota, m_Y,  \eta_Y, A_Y\}$.  But by the same universal reasoning, $\text{Mor}(w^{\otimes m}, w^{\otimes n})$ is generated by $\{\iota, m_Y,  \eta_Y, A_Y\}$, so we conclude that $\text{Mor}(v^{\otimes m}, v^{\otimes n}) = \varphi(\text{Mor}(u^{\otimes m}, u^{\otimes n})) \subseteq \text{Mor}(w^{\otimes m}, w^{\otimes n})$.  

Finally, it remains to show that $X \cong_{A^*} Y$.  Since we have a monoidal equivalence $\varphi:\text{Rep}(G_X) \to \text{Rep}(G_Y)$, Theorem \ref{thm-mon-bg} guarantees the existence of an $\cO(G_Y)$-$\cO(G_X)$-bigalois extension $Z$.  Moreover, from \cite[Theorem 2.3.11]{NeTu13}, one can construct a unitary operator $z \in Z \otimes B(L^2(X),L^2(Y))$ satisfying the relations
\begin{enumerate}
\item \label{one} $(1 \otimes A_Y)z = (1 \otimes \varphi(A_X))z = z(1 \otimes A_X)$. 
\item \label{two} $1 \otimes \eta_Y = 1 \otimes \varphi(\eta_X)  = z(1 \otimes \eta_X)$. 
\item \label{three} $(1 \otimes m_Y)z_{12}z_{13}= (1 \otimes \varphi(m_X))z_{12}z_{13}=  z(1 \otimes m_X)$. 
\item \label{four} $(z^*)_{12}(1 \otimes m_Y^*\eta_Y \otimes 1) =(z^*)_{12}(1 \otimes \varphi(m_X^*\eta_X) \otimes 1) = z_{13}(1 \otimes m_X^*\eta_X)$.
\end{enumerate}  
These four relations say precisely that the map $e_i \mapsto \sum_j f_j \otimes z_{ji}$ defines a unital $\ast$-homomorphism $\cO(X) \to \cO(Y) \otimes Z$ (where $(e_i)$ and $(f_j)$ are ONBs for $L^2(X)$ and $L^2(Y)$).  In particular, we obtain  a non-zero $\ast$-homomorphism $\cO(G_Y,G_X) \to Z$ given by $p \mapsto z$ (where $p$ denotes the matrix of generators of $\cO(G_Y,G_X)$).  I.e., $\cO(G_Y,G_X) \ne 0$. 
\end{proof}

\begin{remark}
With a little more work one can show that in fact $\cO(G_Y,G_X) \cong Z$ via the above homomorphism.
\end{remark}

Theorem \ref{stability-monoidal} supplies us with many easy examples of quantum isomorphic quantum graphs.

\begin{example}
Let $\delta > 0$ and let $X$ and $Y$ be quantum sets each equipped with $\delta$-forms.  Then it follows from \cite[Theorem 4.7]{DeVa10} that the quantum automorphism groups of the spaces $X$ and $Y$ are monoidally equivalent.  In view of Theorem \ref{stability-monoidal}, this is equivalent to saying that the quantum complete graphs $K_X$ and $K_Y$ are C$^\ast$-quantum isomorphic.  In particular,
\begin{itemize}
\item For each $n \ge 4$, we have $K_{n^2} \cong_{qc} K_{X_n}$, where $K_{X_n}$ is the   quantum complete graph associated to the quantum set $X_n = (M_n(\bC), n^{-1}\text{Tr}(\cdot))$.
\item  Let $Q \in M_n(\bC)$ with $Q>0$, $\text{Tr}(Q) = 1$, $\text{Tr}(Q^{-1}) = \delta^2 > 0$, and consider the quantum set $Y = (M_n(\bC), \psi_Y = \text{Tr}(Q \cdot), \delta^2 \psi_Y(\cdot )1)$.  Then $K_Y \cong_{C^*}K_X$ for any quantum set $X$ equipped with a $\delta$-form.
\end{itemize}  
In particular, quantum isomorphic quantum graphs need not have the same dimension. 
\end{example}

% % \begin{definition}
% % Let $X = (\cO(X), \psi_X, A_X)$ and $Y = (\cO(Y),\psi_Y, A_Y)$ be quantum graphs of dimension $d_X$ $d_Y$, respectively.   Define $\cO(G_Y,G_X)$ to be the unital $\ast$-algebra with generators $(z_{ij})_{\substack{1\le i \le d_Y\\1 \le j \le d_X}}$ satisfying the following relations:
% % \begin{enumerate}
% % \item $z= [z_{ij}] \in M_{d_Y,d_X}(\cO(G_Y,G_X))$ is unitary.  I.e., $z^*z = zz^* = 1$.
% % \item $(\iota \otimes m_Y)(z_{12}z_{13}) = z(\iota \otimes m_X)$.
% % \item $z(\iota \otimes \eta_X) = (\iota \otimes \eta_Y)$.
% % \item $(\iota \otimes A_Y)z= z(\iota \otimes A_X)$.
% % \end{enumerate}
% % \end{definition}
% % 
% % \begin{theorem}
% % Let $X = (\cO(X), \psi_X, A_X)$ and $Y = (\cO(Y),\psi_Y, A_Y)$ be quantum graphs and let $G_X, G_Y$ be their corresponding quantum automorphism groups.  Then the following are equivalent.
% % \begin{enumerate}
% % \item $\cO(G_Y,G_X)$ is non-zero. 
% % \item There is a monoidal equivalence $\varphi:\text{Rep}(G_X) \to \text{Rep}(G_Y)$ such that $m_Y = \varphi(m_X)$, $\eta_Y = \varphi(\eta_X)$, and $A_Y = \varphi(A_X)$. 
% % \end{enumerate}
% % In particular, $\cO(G_Y,G_X) \ne 0$ if and only if it admits a faithful state.
% % \end{theorem}
% %

Let $X$ and $Y$ be two quantum graphs. We noted above, in the discussion preceding \Cref{stability-monoidal}, that a monoidal equivalence between $G_X$ and $G_Y$ sending $A_X$ to $A_Y$ would by necessity force the graphs to be isospectral. In particular, they must be so if they are quantum isomorphic. We end this section with an example of a pair of non-quantum-isomorphic, isospectral graphs with trivial quantum automorphism groups.

We will use the Frucht graph $X$, which is a $3$-regular graph on $12$ vertices and has trivial automorphism group. Moreover, its adjacency matrix has no repeated eigenvalues; as we will see, this will eventually help show that $X$ has no {\it quantum} automorphisms.

Before delving into the statement and proof of the next result we fix some notation and terminology. $X$ will be a (classical) graph on $n$ vertices, and we denote by $p_i$, $1\le i\le n$ the minimal projections of $\cO(X)$. The {\it support} of an element $f\in \cO(X)$ is the set of $p_i$ that have non-zero coefficients in a decomposition of $f$ as a linear combination
\begin{equation*}
  f=\sum_{i=1}^n \alpha_i p_i. 
\end{equation*}

\begin{lemma}\label{le.simple}
  Let $X$ be a graph on $n$ vertices whose adjacency matrix $A=A_X$ has only simple eigenvalues with no two eigenvectors having disjoint supports. Then the quantum automorphism group $G_X$ is classical.   
\end{lemma}
\begin{proof}
  As a consequence of \cite[Theorem 3.11]{Lu17}, $A$ is contained in the space of endomorphisms of $\cO(X)$ regarded as an $\cO(G_X)$-comodule. In particular the elements $e_i$, $1\le i\le n$ of an $A$-eigenbasis of $\cO(X)$ span lines invariant under $\cO(G_X)$, meaning that the coaction takes the form
  \begin{equation*}
    \rho_X:e_i\mapsto e_i\otimes g_i
  \end{equation*}
  for elements $g_i\in \cO(G_X)$. It follows from this that the elements $g_i$ are {\it grouplike}, i.e.
  \begin{equation*}
    \Delta(g_i) = g_i\otimes g_i,\quad \varepsilon(g_i)=1.
  \end{equation*}
  Since $\cO(G_X)$ is generated as an algebra by the right hand tensorands of $\rho_X(e_i)$, it follows that $\cO(G_X)$ is the group algebra $\bC\Gamma$ of a group $\Gamma$ (generated by $g_i$) and $\rho_X$ is nothing but a $\Gamma$-graded algebra structure on $\cO(X)$. In order to conclude it remains to argue that the grading group $\Gamma$ is commutative, for it will then follow that $G_X$ is the (classical) Pontryagin dual of $\Gamma$.

  Now, since no two $e_i$ have disjoint supports, $e_ie_j=e_je_i$ are all non-zero and homogeneous of degree $g_ig_j$ as well as $g_jg_i$. It follows that the generators $g_i$ of $\Gamma$ all commute, and we are done.
\end{proof}

\begin{corollary}\label{cor.frcht-triv}
  The Frucht graph has trivial quantum automorphism group. 
\end{corollary}
\begin{proof}
  The graph meets the requirements of \Cref{le.simple}: we have verified the no-disjoint-supports condition by direct examination, having computed the entries of the eigenvectors to three decimals using a CAS.

  It follows from \Cref{le.simple} that the quantum automorphism group is classical. On the other hand, we know that the Frucht graph has no non-trivial {\it classical} automorphisms.
\end{proof}

We also need the following simple remark

\begin{lemma}\label{le.link}
Let $X_i$, $i=1,2$ be graphs with the vertex set $\{1,\dots,n\}$ and respective adjacency matrices $A_i$. Suppose that $P=(p_{ij})$ is a magic unitary such that $A_1P=PA_2$. If $p_{ij}\neq 0$ then $\mathrm{deg}(i) = \mathrm{deg}(j)$. 
\end{lemma}
\begin{proof}
The $(i,j)$ entries of both sides are equal, so $\sum_{k} a_{ik} p_{kj} = \sum_{k} p_{ik} b_{kj}$. We can sum both sides of equality, getting $\sum_{k} \left( \sum_{i} a_{ik}\right) p_{kj} = \left(\sum_{k} b_{kj}\right) \mathrm{Id}$. Since $(p_{kj})_{k}$ is a partition of unity, this can only happen if $\sum_{i} a_{ik} = \sum_{k} b_{kj}$, whenever $p_{kj}$ is non-zero. But the left-hand side is $\deg_{A}(k)$ and the right-hand side is $\deg_{B}(j)$, so we are done. 
\end{proof}

We now need the following consequence of \cite[Theorem 2.2]{gm}.

\begin{theorem}
  Suppose that $\Gamma$ is a regular graph with $2m$ vertices. Form $\Gamma'$ by adding a vertex $2m+1$, which is joined to exactly $m$ vertices. If $\Gamma''$ is the switch of $\Gamma'$, i.e. the graph formed by connecting $2m+1$ to the other $m$ vertices of $\Gamma$, then $\Gamma'$ and $\Gamma''$ are isospectral.
\end{theorem}

\begin{example}\label{ex.niso}
  We construct the pair $X_i$, $i=1,2$ of graphs alluded to above as follows.

Start with the Frucht graph $X$, which is $3$-regular and has $12$ vertices. Its adjacency matrix has simple eigenvalues, hence its quantum automorphism group is trivial by \Cref{cor.frcht-triv}. Next, form the $13$-vertex graph $X_1$ obtained by adding a vertex connected to exactly $6$ vertices of the Frucht graph. Note first that $X_1$ once more has a trivial quantum automorphism group. Indeed, the degrees of the vertices of $X$ are all $3$ while the additional vertex has degree $6$. By \Cref{le.link} any the magic unitary commuting with the adjacency matrix would be block-diagonal, consisting of a $12\times 12$ block and the unit in position $(13,13)$. The former block would however constitute a quantum automorphism of the Frucht graph itself, so that block must be trivial.

The same reasoning shows that the graph $X_2$ obtained by connecting the $13^{th}$ vertex of $X_1$ to the {\it other} $6$ vertices of $X$ cannot be quantum isomorphic to $X_1$. As claimed before, $X_i$ are isospectral, non-quantum-isomorphic, and have trivial quantum automorphism groups.
\end{example}

Incidentally, \Cref{ex.niso} also also answers a question posed implicitly in \cite{Lu17}. As mentioned before in the discussion preceding \Cref{subse.states}, \cite[Theorem 4.2]{Lu17} proves that for quantum isomorphic graphs $X_i$, $i=1,2$ there is an isomorphism between their respective quantum orbital algebras that identifies the adjacency matrices $A_i$, $i=1,2$. The discussion following that result mentions that while the converse is not expected to hold, the authors do not have a counterexample. The graphs $X_i$ constructed above provide that counterexample:

Since the quantum automorphism groups $G_i$, $i=1,2$ are trivial the quantum orbital algebras are simply the matrix algebras $\mathrm{End}(\cO(X_i))$. Since $A_i$ are isospectral self-adjoint matrices, there is an isomorphism
\begin{equation*}
\mathrm{End}(\cO(X_1))  \cong\mathrm{End}(\cO(X_2))
\end{equation*}
identifying them. 

\section{Applications to other synchronous games}\label{se.appl}

Theorem~\ref{th.cls-grph} can be applied to obtain results about families of games. The key concepts that we need to accomplish this are a notion of equivalence for games and the concept of a {\it hereditary $*$-algebra.}

\begin{definition}  A $*$-algebra $\cA$ is called {\bf hereditary} provided that, whenever $n \in \bN$ and $x_1,...,x_n \in \cA$ are such that $\sum_{i=1}^n x_i^*x_i=0$, then $x_i=0$ for all $1 \leq i \leq n$.
\end{definition}

One key advantage of hereditary $*$-algebras is that if $\cA$ is a hereditary $*$-algebra and we set 
\[ \cP= \left\{ x \in \cA: \exists x_1,...,x_n \in \cA \text{ such that } x= \sum_{i=1}^n x_i^*x_i \right\}, \]
then $\cP \cap (- \cP) = \{0\}$.  Thus, we may define, for $a=a^*$ and $b= b^*$, a partial order by 
\[a \le b \iff  \exists x_1,...,x_n \in \cA \text{ such that } b-a = \sum_i x_i^*x_i.\]
We note that if $a \le b$ and $b \le a$, then $a=b$.

Given a $*$-algebra $\cA$, the smallest two-sided, $*$-closed hereditary ideal $\cI$ containing $0$ is called the {\it hereditary kernel} of $\cA$, and the quotient $\cA/\cI$ is denoted by $\cA_{hered}$. Given a synchronous game $\cG$, we let $\cA_{hered}(\cG)$ denote the hereditary quotient of $\cA(\cG)$.  Note that by Theorem \ref{th.cls-grph}, $\cA(ISo(X,Y))$ admits a faithful tracial state whenever $\cA(Iso(X,Y)) \ne 0$, and therefore $\cA(Iso(X,Y))  = \cA_{hered}(Iso(X,Y))$.

Note that if $\cA$ and $\cB$ are $*$-algebras with $\cB$ hereditary and $\pi: \cA \to \cB$ is a $*$-homomorphism, then the kernel of $\pi$ contains the hereditary kernel of $\cA$ and so induces a $*$-homomorphism $\tilde{\pi}: \cA_{hered} \to \cB$. Hence, for any pair of $*$-algebras $\cA$ and $\cB$, every $*$-homomorphism $\pi: \cA \to \cB$ induces a $*$-homomorphism $\tilde{\pi}: \cA_{hered} \to \cB_{hered}$. 

\begin{definition} Let $\cl G_1$ and $\cl G_2$ be two synchronous games.  We say that $\cl G_1$ and $\cl G_2$ are {\bf $\ast$-equivalent} if there exist unital $\ast$-homomorphisms $\pi: \cl A(\cl G_1) \to \cl A(\cl G_2)$ and $\rho: \cl A(\cl G_2) \to \cl A(\cl G_1)$. We say that $\cG_1$ and $\cG_2$ are {\bf hereditarily $\ast$-equivalent} if there exist unital $\ast$-homomorphisms $\pi: \cA_{hered}(\cG_1) \to \cA_{hered}(\cG_2)$ and $\rho: \cA_{hered}(\cG_2) \to \cA_{hered}(\cG_1)$.
\end{definition}
We allow the possibility that one of the two algebras is $(0)$, in which case $1=0$ in that algebra. In this case, equivalence of the algebras implies that the other algebra is also $(0)$.

Note that we do not require $\pi$ and $\rho$ to be mutual inverses or even one-to-one, just unital. The reason for examining this relation is given below.

\begin{proposition} Let $t \in \{ loc, q, qa, qc, C^*\}$.  If $\cl G_1, \cl G_2$ are synchronous games that are hereditarily $\ast$-equivalent, then $\cl G_1$ has a perfect $t$-strategy if and only if $\cl G_2$ has a perfect $t$-strategy. If, in addition, the games $\cG_1$ and $\cG_2$ are $*$-equivalent, then $\cG_1$ has a perfect $A^*$-strategy if and only if $\cG_2$ has a perfect $A^*$-strategy.
\end{proposition}

\begin{proof} We do the case $t=q$, the rest are similar.  First assume that the algebras are $*$-equivalent. If $\cl G_2$ has a perfect $q$-strategy, then there is a unital $\ast$-morphism $\gamma: \cl A(\cl G_2) \to M_d$ for some $d$. Composing with $\pi$ yields a $\ast$-homomorphism from $\cl A(\cl G_1)$ into $M_d$, and so, $\cl G_1$ has a perfect $q$-strategy. Since $M_d$ is a hereditary $*$-algebra, the same reasoning applies when the algebras are hereditarily $*$-equivalent. The converse is clear, as are the remaining cases.
\end{proof}

We now introduce another game which we will show is hereditarily $\ast$-equivalent to a graph isomorphism game.

%%%%%%%%%%%%%%%%%%%%%%%%%%%%%%%%%%%%%%%%%%%%%%%%%%%%%%%%%%%%%%%%%%%%%%%%%%%%%
\subsection{The syncBCS game}
This game was first introduced in \cite{KPS18} and is a synchronous version of what is classically known as the BCS game.  Given an $m \times n$ matrix $A=(a_{i,j})$ over the field of two elements, $\mathbb Z_2$ and a vector $b$, we introduce a game denoted $syncBCS(A,b)$, that is intended to convince a referee that Alice and Bob have a solution $x$ to the equation $Ax=b$.

For $i = 1, \ldots, m$, let $V_i = \{ j: a_{i,j} \ne 0 \}$. Note that to solve the $i$-th equation in $Ax=b$, we only need
\[ \sum_{j \in V_i} a_{i,j} x_j = b_i,\]
since the remaining terms are irrelevant. Set
\[ S_i^b = \{ x \in \mathbb Z_2^n : \sum_{j \in V_i} a_{i,j} x_j = b_i \text{ and } x_j = 0 \text{ for  }j \notin V_i \}. \]
We associate a synchronous game to $Ax = b$ as follows:

\begin{definition}
Suppose $Ax = b$ is an $m \times n$ linear system over $\mathbb{Z}_2$ and $b \in \mathbb{Z}_2^n$.
The synchronous BCS game associated to $A x = b$, denoted $synBCS(A, b)$, is given as follows:
\begin{enumerate}
  \item the input set is $\mathcal{I} = \{1,\ldots, m\}$;
  \item the output set is $\mathcal{O} = \mathbb Z_2^n$;
  \item given input $(i,j)$, Alice and Bob win on output $(x,y)$ if and only if $x \in S_i^b$, $y \in S_j^b$, and for all $k \in V_i \cap V_j$, $x_k = y_k$.
\end{enumerate}
\end{definition}

Next let us recall from \cite[Section 6]{amrssv} the graph $G_{A, b}$ defined for a linear system $Ax = b$ over $\mathbb{Z}_2$.

% % reference was \cite{graphisom} before

\begin{definition} Suppose $Ax = b$ is an $m \times n$ linear system over $\mathbb{Z}_2$ and $b \in \mathbb{Z}_2^n$.
  Define a graph $G_{A, b}$ with the following data:
  \begin{enumerate}
    \item the vertices of $G_{A, b}$ are pairs $(i,x)$ where $i \in \{1,\ldots, m\}$ and $x \in S_i^b$;
    \item there is an edge between distinct vertices $(i, x)$ and $(j, y)$ if and only if there exists some $k \in V_i \cap V_j$ for which $x_k \neq y_k$; that is, $x$ and $y$ are inconsistent solutions.
  \end{enumerate}
\end{definition}

We are now ready to state the main theorem of this section.

\begin{theorem} Let $A=(a_{i,j})$ be an $m \times n$ matrix over $\mathbb Z_2$ and let $b \in \mathbb Z_2^n$. Then the following three synchronous games:
\begin{enumerate}
\item  $syncBCS(A,b)$,
\item $Iso( G_{A,b}, G_{A,0})$,
\item $Hom(K_m, \overline{G_{A,b}})$,
\end{enumerate}
are hereditarily $\ast$-equivalent.
\end{theorem}

Before proving this theorem, we state some corollaries. Combining the above theorem with Theorem~\ref{th.cls-grph} yields the following consequences.

\begin{corollary} \label{vs-app1} Let $A=(a_{i,j})$ be an $m \times n$ matrix over $\mathbb Z_2$ and let $b \in \mathbb Z_2^n$.
The following are equivalent:
\begin{enumerate}
\item $\cA_{hered}(syncBCS(A,b)) \ne (0)$,
\item $syncBCS(A,b)$ has a perfect C$^\ast$-strategy,
\item $syncBCS(A,b)$ has a perfect qc-strategy.
\end{enumerate}
\end{corollary}

\begin{corollary} \label{vs-app2} Let $A=(a_{i,j})$ be an $m \times n$ matrix over $\mathbb Z_2$ and let $b \in \mathbb Z_2^n$. The following are equivalent:
\begin{enumerate}
\item $\cA_{hered}(Hom(K_m, \overline{G_{A,b}})) \ne (0)$,
\item $Hom(K_m, \overline{G_{A,b}})$ has a perfect C$^\ast$-strategy,
\item $Hom(K_m, \overline{G_{A,b}})$ has a perfect qc-strategy.
\end{enumerate}
\end{corollary}

The proof of the theorem borrows some ideas from the proof of \cite[Theorem~5.4]{KPS18}.
\\

\begin{proof}
We begin by constructing a unital $\ast$-homomorphism from $\cl A(Iso(G_{A,b}, G_{A,0}))$ to $\cl A(syncBCS(A,b))$. By our earlier remarks, this $\ast$-homomorphism will induce a unital $\ast$-homomorphism from 

$\cA_{hered}(Iso(G_{A,b}, G_{A,0}))$ to $\cA_{hered}(syncBCS(A,b))$.

The algebra $\cl A(syncBCS(A,b))$ is generated by projections $e_{i, x}$ for $i = 1, \ldots, m$ and $x \in \mathbb Z_2^n$ satisfying $\sum_x e_{i, x} = 1$ for all $i$, $e_{i,x}e_{i,y}=0$ if $x \ne y$. Moreover, given input $i$, if $x \notin S_i^b$, then they lose for all $(j,y)$, from this it follows that  $e_{i, x} = 0$ if $x \notin S_i^b$. Also, if $x \in S_i^b$ and $y \in S_j^b$, then $e_{i, x} e_{j, y} = 0$ if there is a $k \in V_i \cap V_j$ with $x_k \neq y_k$.

Let $S_i^0 \subseteq \mathbb Z_2^n$ denote the set of solutions to the $i$th equation of the linear system $Ax = 0$ and let $S_i^b \subseteq \mathbb Z_2^n$ denote the set of solutions to the $i$th equation of the linear system $Ax = b$.   Note that if $y \in S_i^0$ and $x \in S_i^b$, then $x+y \in S_i^b$.  Moreover, for $x \in S_i^b$, the map $S_i^0 \rightarrow S_i^b$ given by $y \mapsto x+y$ is a bijection.

The algebra $\cl A(Iso(G_{A,b}, G_{A,0}))$ is generated by projections  $e_{(i,x),(j,y)}$ with $(i,x) \in V(G_{A,b})$ and $(j,y) \in V(G_{A,0})$, satisfying certain relations.
For $(i, x) \in V(G_{A, b})$ and $(j, y) \in V(G_{A, 0})$, define
\[ q_{(i, x), (j, y)} = \begin{cases} e_{i, x+y} & i = j \\ 0 & i \neq j \end{cases} \]
and note that each $q_{(i, x), (j, y)}$ is a projection.  For $(i, x) \in V(G_{A, b})$, we have
\[ \sum_{(j, y) \in V(G_{A, 0})} q_{(i, x), (j, y)} = \sum_{j=1}^n \sum_{y \in S_j^0} q_{(i, x), (j, y)} = \sum_{y \in S_i^0} e_{i, x+y} = \sum_{z \in S_i^b} e_{i, z} = 1. \]
A similar computation shows that for all $(j, y) \in V(G_{A, 0})$, we have
\[ \sum_{(i, x) \in V(G_{A, b})} q_{(i, x), (j, y)} = 1. \]

We need to show that for all $(i, x), (i', x') \in V(G_{A, b})$ and $(j, y), (j', y') \in V(G_{A, 0})$, the implication
\[ q_{(i,x),(j,y)} q_{(i',x'),(j',y')} \neq 0 \quad \Rightarrow \quad \operatorname{rel}((i,x), (i',x')) = \operatorname{rel}((j,y),(j',y')) \]
holds.  To this end, suppose $q_{(i,x),(j,y)} q_{(i',x'),(j',y')} \neq 0$.  Then $i = j$, $i' = j'$, and $e_{i, x+y} e_{i', x'+y'} \neq 0$.  We consider several cases.

Suppose first $i = i'$.  Then we have $x+y = x'+y'$.  If $x = x'$, then $y = y'$ and we have both $(i, x) = (i', x')$ and $(j, y) = (j', y')$ so the right hand side of the implication holds in this case.  Conversely, if $x \neq x'$ and $y \neq y'$, then $(i, x) \neq (i', x')$ and $(j, y) \neq (j', y')$.  Note also that since $i = i'$, $x$ and $x'$ are necessarily inconsistent solutions so that $(i, x)$ and $(i', x')$ are adjacent.  Similar reasoning shows $(j, y)$ and $(j', y')$ are adjacent.  Hence the right hand side of the implication holds.

Now assume $i \neq i'$ so that, in particular, $(i, x) \neq (i', x')$.  If $(i, x)$ and $(i', x')$ are adjacent, there is a $k \in V_i \cap V_{i'}$ such that $x_k \neq x'_k$.  On the other hand, as $e_{i, x+y} e_{i', x'+y'} \neq 0$, we know $x_k +y_k = x'_k +y'_k$.  Therefore, $y_k \neq y'_k$ so that $(i, y)$ and $(i', y')$ are adjacent.  Finally, suppose $(i, x)$ and $(i', x')$ are not adjacent.  Then $x_k = x'_k$ for all $i \in V_i \cap V_{i'}$.  Again since $e_{i, x+y} e_{i', x'+y'} \neq 0$, we also know $x_k+ y_k = x'_k +y'_k$ for all $k \in V_i \cap V_{i'}$ and therefore $y_k = y'_k$ for all $k \in V_i \cap V_{i'}$ so that $(j, y)$ and $(j', y')$ are not adjacent.  This covers all cases.

Now, by the fact that $\cl A(Iso(G_{A,b}, G_{A,0}))$ is the universal $\ast$-algebra with projections satisfying these properties, we have that the map  $e_{(i,x),(j,y)} \to q_{(i, x), (j, y)} \in \cl A(syncBCS(A,b))$ defines the desired unital $\ast$-homomorphism.

Now we prove that there is a unital $\ast$-homomorphism from $\cl A(Hom(K_m, \overline{G_{A,b}}))$ to $\cl A(Iso(G_{A,b}, G_{A,0}))$.
Note that for any graph $X$ we have that $\cl A(Hom(K_m, X))$ is generated by projections, $e_{i,x}, \, 1 \le i \le m, \, x \in V(X)$ satisfying  $\sum_x e_{i,x} =1, \, e_{i,x}e_{i,y} =0, \, x \ne y$ and $i \ne j, (x,y) \notin E(X) \implies e_{i,x} e_{j,y} =0.$ Since we are interested in $Hom(K_m, \overline{X})$, this last relation changes to $i \ne j, (x,y) \in E(X) \implies e_{i,x}e_{j,y}=0$. For each $(i,x) \in V(G_{A,b})$ and $1 \le j \le m$ we define an element $p_{j, (i,x)} \in \cl A(Iso( G_{A,b}, G_{A,0}))$ by setting $p_{j, (i,x)} = e_{(i,x), (j,0)}$.
We have that  $\sum_{(i,x) \in V(G_{A,b})} p_{j, (i,x)} =1$ and $p_{j, (i,x)} p_{j, (i^{\prime}, x^{\prime})} =0$ when $(i,x) \ne (i^{\prime}, x^{\prime})$ by the magic permutation relations.

Finally, if $j \ne l$ and $((i,x), (i^{\prime}, x^{\prime})) \in E(G_{A,b})$ then $rel((j,0),(l,0)) = +1$ while $rel((i,x), (i^{\prime}, x^{\prime})) = -1$.  Hence,
\[ p_{j,(i,x)}p_{l,(i^{\prime}, x^{\prime})} = e_{(i,x), (j,0)} e_{(i^{\prime},x^{\prime}), (l,0)} = 0.\]
This shows that the map from $\cl A(Hom(K_m, \overline{G_{A,b}}))$ to $\cl A(Iso(G_{A,b}, G_{A,0}))$ given by $e_{j, (i,x)} \to p_{j, (i,x)}$ defines a unital $\ast$-homomorphism and again this will induce a unital *-homomorphism between their hereditary quotients.

Finally, we must exhibit a unital $\ast$-homomorphism from $\cl A(syncBCS(A,b))$ into $\cl A_{hered}(Hom(K_m, \overline{G_{A,b})}).$

This latter algebra is generated by projections $e_{i, (j,x)}$, $1 \le i \le m$, $(j,x) \in V(G_{A,b})$, i.e., $x \in S_j^b$. These satisfy $\sum_{j,x} e_{i, (j,x)} =1$ for all $i$, and $e_{i, (j,x)}e_{i,(k,y)} =0$ whenever $(j,x) \ne (k,y)$. Moreover, since $(i,l)$ is an edge in $K_m$ whenever $i \neq l$, we have that when $i \ne l$ and $((j,x),(k,y))$ is not an edge in $\overline{G_{A,b}}$ (meaning that $x \in S_j^b$ and $y \in S_k^b$ are inconsistent solutions), then  $e_{i, (j,x)} e_{l, (k,y)} =0$.

Note that if $x,y \in S_i^b$ and $x \ne y$, then $e_{k,(i,x)}e_{k,(i,y)}=0$. If $k \ne j$, then $k$ and $j$ are connected by an edge in $K_m$, while $(i,x)$ and $(i,y)$ are not connected by an edge in $\overline{G_{A,b}}$, so that $e_{k,(i,x)}e_{j,(i,y)}=0$.  From these facts, it follows that
\[p_i := \sum_{k=1}^m \sum_{x \in S_i^b} e_{k,(i,x)}\]
is a self-adjoint idempotent. Set $q_i = 1- p_i = q_i^2$. Then
\[ \sum_{k=1}^m q_i^2 = \sum_{k=1}^m (1 - p_i) = m \cdot 1 - \sum_{k=1}^m \sum_{i=1}^m \sum_{x \in S_i^b} e_{k,(i,x)} = 0,\]
using the fact that $\sum_{j,x} e_{i,(j,x)}=1$ for all $i$. Thus, we have that
\[q_i=0 \text{ and } p_i =1, \, \forall 1 \le i \le m.\]

For $x \in S_i^b$, set
\[ f_{i,x} = \sum_{k=1}^m e_{k,(i,x)}.\]
Then $f_{i,x}= f_{i,x}^*$ and for $k \ne j$, we have that $k,j$ are connected by an edge in $K_m$, while $(i,x)$ is not connected to $(i,x)$ by an edge; hence,
\[ f_{i,x}^2 = \sum_{k,j=1}^m e_{k,(i,x)}e_{j,(i,x)} = \sum_{k=1}^m e_{k,(i,x)} = f_{i,x}, \]
so that $f_{i,x}$ is a self-adjoint idempotent.
Also, for $x,y \in S_i^b$ with $x \ne y$, we have that
\[ f_{i,x} f_{i,y} = \sum_{j,k=1}^m e_{k,(i,x)} e_{j,(i,y)} = \sum_{k=1}^m e_{k,(i,x)}e_{k,(i,y)} =0,\]
and
\[ \sum_{x \in S_i^b} f_{i,x} = \sum_{k=1}^m \sum_{x \in S_i^b} e_{k,(i,x)} = p_i =1.\]
Thus, for each $i$, $\{ f_{i,x}: x \in S_i^b \}$ is a set of self-adjoint idempotents whose sum is $1$.

Finally, if $(i,x)$ and $(j,y)$ are inconsistent solutions, then
$$ f_{i,x} f_{j,y} = \sum_{k,h=1}^m e_{k,(i,x)}e_{h,(j,y)}. $$
When $h=k$, each of these products is 0. For $h \ne k$, we have that $h$ and $k$ are connected by an edge in $K_m$ and so the product will be 0, since $x$ and $y$ being inconsistent solutions implies that $(i,x)$ and $(j,y)$ are not connected by an edge in $\overline{G_{A,b}}$.

Thus, the set $\{ f_{i,x} \}$ satisfies the relations on the generators of the free algebra $\cA(syncBCS(A,b))$ and they induce a unital $*$-homomorphism from $\cA(syncBCS(A,b))$ into $\cA_{hered}(Hom(K_m, \overline{G_{A,b}}))$, from which the result follows.
\end{proof}

\begin{remark} In an earlier version of this paper we claimed that the three games were $*$-equivalent.  We now know that this is incorrect. In fact, it is possible to construct linear systems for which
\[ \cA(syncBCS(A,b)) = \cA(Iso(G_{A,b}, G_{A,0})) = (0),\]
while  $\cA(Hom(K_m, \overline{G_{A,b}})) \ne (0).$
It would be interesting to know whether or not $syncBCS(A,b)$ and $Iso(G_{A,b}, G_{A,0})$ are $\ast$-equivalent.
\end{remark}

\subsection*{Acknowledgements}

MB, KE and XS were supported by NSF grant DMS-1700267.  AC was supported by NSF grant DMS-1801011. VP and SH were supported by NSERC. MW was supported by the National Science Centre (NCN) grant 2016/21/N/ST1/02499, European Research Council Consolidator Grant
614195 RIGIDITY and by long term structural funding - Methusalem grant of the Flemish Government. Parts of this work were completed during the 2017 CMS Winter Meeting at the University of Waterloo, a research visit of MW to Texas A\&M University (supported by National Science Centre (NCN) scholarship 2017/24/T/ST1/00391) and during the Fall 2018 Workshop on ``Cohomology of quantum groups and quantum automorphism groups of finite graphs'' at the University of Saabr\"ucken.  The authors thank the Universities and their respective hosts for the very  stimulating research environments.

%%%%%%%%%%%%%%%%%%%%%%%%%%%%%%%%%%%%%%%%%%%%%%%%%%%%%%%%%%%%%%%%%%%%%%%%%%%%%%%%%%%%%%%%%%%%%%%%%%%%%%%%%%%%%%%%%%
%%%%%%%%%%%%%%%%%%%%%%%%%%%%%%%%%%%%%%%%%%%%%%%%%%%%%%%%%%%%%%%%%%%%%%%%%%%%%%%%%%%%%%%%%%%%%%%%%%%%%%%%%%%%%%%%%%

%%%%%%%%%%%%%%%%%%%%%%%%%%%%%%%%%%%%%%%%%%%%%%%%%%%%%%%%%%%%%%%%%%%%%%%%%%%%%%%%%%%%%%%%%%%%%%%%%%%%%%%%%%%%%%%%%%

%%%%%%%%%%%%%%%%%%%%%%%%%%%%%%%%%%%%%%%%%%%%%%%%%%%%%%%%%%%%

%%%%%%%%%%%%%%%%%%%%%%%%%%%%%%%%%%%%%%%%%%%%%%%%%%%%%%%%%%%%%%%%%%%%%%%%%%%%%%%%%%%%%%%%%%%%%%%%%%%%%%%%%%%%%%%%%%
%%%%%%%%%%%%%%%%%%%%%%%%%%%%%%%%%%%%%%%%%%%%%%%%%%%%%%%%%%%%%%%%%%%%%%%%%%%%%%%%%%%%%%%%%%%%%%%%%%%%%%%%%%%%%%%%%%

% \bibliography{q-iso-monoidal}{}
\bibliographystyle{plain}
\addcontentsline{toc}{section}{References}

\def\polhk#1{\setbox0=\hbox{#1}{\ooalign{\hidewidth
  \lower1.5ex\hbox{`}\hidewidth\crcr\unhbox0}}}

%\Addresses

\end{document}